\documentclass[journal]{IEEEtran}

\ifCLASSINFOpdf
  \usepackage[pdftex]{graphicx}
  \graphicspath{{../pdf/}{../jpeg/}}
  \DeclareGraphicsExtensions{.pdf,.jpeg,.png}
\else
  \usepackage[dvips]{graphicx}
\fi

\usepackage{cite}
\usepackage[utf8]{inputenc}
\usepackage{todonotes}
\usepackage{amssymb}
\usepackage{amsmath}
\usepackage{amsthm}
\usepackage{mathtools}
\usepackage{optidef}
\usepackage{hyperref}
\usepackage{algorithm}
\usepackage{algpseudocode}
\usepackage{graphicx}
\usepackage{subfig}
\usepackage{booktabs}
\usepackage{cuted}
\usepackage{flushend}
\usepackage{tikz}
\usepackage{textcomp}
\usepackage{hyperref}
\usepackage{lipsum}
\usepackage{eso-pic}
\newcommand\copyrighttext{%
  \footnotesize \textcopyright 2019 IEEE. Personal use of this material is permitted.  Permission from IEEE must be obtained for all other uses, in any current or future media, including reprinting/republishing this material for advertising or promotional purposes, creating new collective works, for resale or redistribution to servers or lists, or reuse of any copyrighted component of this work in other works.}
\newcommand\copyrightnotice{%
\begin{tikzpicture}[remember picture,overlay]
\node[anchor=south,yshift=0pt] at (current page.south) {\fbox{\parbox{\dimexpr\textwidth-\fboxsep-\fboxrule\relax}{\copyrighttext}}};
\end{tikzpicture}%
}

\newtheorem{theorem}{Theorem}
\newtheorem{corollary}{Corollary}

\begin{document}

\title{Assignment and Pricing of Shared Rides in Ride-Sourcing using Combinatorial Double Auctions}

\author{Renos~Karamanis, Eleftherios~Anastasiadis, Panagiotis~Angeloudis and~Marc~Stettler
\thanks{R. Karamanis, E. Anastasiadis, P. Angeloudis and M. Stettler are with the Department
of Civil and Environmental Engineering, Imperial College London,
UK, e-mail: renos.karamanis10@imperial.ac.uk.}}

\IEEEtitleabstractindextext{
\begin{abstract}
Transportation Network Companies employ dynamic pricing methods at periods of peak travel to incentivise driver participation and balance supply and demand for rides. 
Surge pricing multipliers are commonly used and are applied following demand and estimates of customer and driver trip valuations. 
Combinatorial double auctions have been identified as a suitable alternative, as they can achieve maximum social welfare in the allocation by relying on customers and drivers stating their valuations.
A shortcoming of current models, however, is that they fail to account for the effects of trip detours that take place in shared trips and their impact on the accuracy of pricing estimates.
To resolve this, we formulate a new shared-ride assignment and pricing algorithm using combinatorial double auctions.
We demonstrate that this model is reduced to a maximum weighted independent set model, which is known to be APX-hard.
A fast local search heuristic is also presented, which is capable of producing results that lie within 10\% of the exact approach for practical implementations. Our proposed algorithm could be used as a fast and reliable assignment and pricing mechanism of ride-sharing requests to vehicles during peak travel times. 
\end{abstract}

}
\IEEEoverridecommandlockouts

\maketitle

\IEEEpubidadjcol
\IEEEdisplaynontitleabstractindextext
\IEEEpeerreviewmaketitle
\copyrightnotice{}

\section{Introduction}
\IEEEPARstart{T}{he} recent proliferation of Transportation Network Companies (TNCs) has been facilitated by an increasing demand for efficient, economic and personalised modes of urban mobility. 
TNCs have quickly captured significant share of the urban mobility market, by providing a service that is usually cheaper than taxis, more convenient than public transport, and an effective alternative to private car ownership. 
Their success has been underpinned by the use of powerful algorithms and analytics, which helped reduce waiting times and increase fleet utilisation \cite{maciejewski2016assignment,Zha2016, doi:10.1080/01441647.2018.1497728}.

To maintain a balance between the supply and demand for rides, TNCs frequently apply dynamic pricing measures \cite{Chen2016} usually taking the form of variable surge tariff multipliers.
Such measures can motivate drivers to attend under-served areas, dampen demand by eliminating requests from riders that are delaying their departure, and also incentivize shared rides between customers or the use of public transport \cite{banerjee2015pricing}.

Through these methods, TNCs effectively operate a two-sided market, to the benefit of both the drivers (supply) and the riders (demand).
As many major TNCs consider to deploy Autonomous Vehicles (AVs) in the near future, their platforms are likely to be transformed into one-sided markets, where they will enjoy complete control of the supply \cite{qiu2018dynamic}. 
Previous work by Karamanis et al. \cite{Karamanis2018} demonstrated that such platforms can still incentivise the use of shared rides or public transport while remaining profitable.

Existing dynamic pricing methods suggest new equilibrium prices to customers without having prior knowledge of their trip valuations. 
If these are considered, market equilibrium prices could be identified without approximation, therefore transforming this process into an auction. 
Previous work on ride pricing using auction theory \cite{Zhang2016, lam2016combinatorial, Asghari2017, Yu2018} focused on the interactions between riders (bidders) and drivers (sellers) who are expected to declare their valuations and costs for prospective rides. 

A TNC platform, taking the role of the auctioneer, would be responsible for determining the winner of each auction \cite{Yu2018}. 
Possible auction settings might involve one or multiple drivers that are assigned to customers sequentially or simultaneously. 
Manipulations of the auctions by either side can be avoided through the use of mechanism design theory, and the analysis of participation incentives. 

Combinatorial Double Auctions (CDAs) \cite{Xia2005} can be used to allocate multiple drivers to riders simultaneously and efficiently\footnote{Economically efficient auction allocations maximise social welfare.} using linear programs that are commonly referred to as winner determination problems (WDP) \cite{lam2016combinatorial}, which are known to be NP-hard.
WDPs can be formulated as set packing problems that maximise the auctioneer's revenue (or social welfare) while taking into account the utilities of the participants \cite{Lehmann2005}. 

Research in dynamic ride-sharing (DRS) (carpooling), is particularly relevant, with several studies exploring the applicability of auction models between commuting drivers and riders \cite{Kleiner2011, Zhao2014, Zhang2016}. 
In \cite{Zhang2016}, the authors propose a CDA discounted trade reduction mechanism for DRS assignment and pricing. 
The proposed mechanism is found to be incentive-compatible \footnote{Incentive-compatible mechanisms ensure that every participant is incentivised to be truthful.}, individually rational\footnote{Individual rationality ensures that no participant incurs a loss. \cite{deVries}} and weakly budget-balanced\footnote{Weakly budget-balanced mechanisms ensure that auctioneer will not incur a loss \cite{Zhao2014}}. 
A system of parallel DRS auctions was proposed in \cite{Kleiner2011} aiming to identify rider-driver matches that minimise detours. 
A DRS model using mechanism design was presented in \cite{Zhao2014}, demonstrating that maximum social welfare cannot be feasibly reached while incentivising the participation of commuters and truthful reporting of trip reservation prices. 

Lam \cite{lam2016combinatorial} models the allocation of AV seats to customers as a combinatorial auction, using the Vickrey-Clarkes-Grove (VCG) mechanism to sequentially assign customers to vehicles and determine prices. 
Three types of service are considered: private rides, shared rides and requests split over multiple vehicles. A separate study developed a CDA model for dial-a-ride AV fleets \cite{Yu2018} where multiple customers and AV operators submit bids,  while a platform determines allocations that maximise social welfare. 
The model is applied for three types of service as in \cite{lam2016combinatorial}, with prices computed using a relaxed version of the problem with Lagrangian multipliers. 
The algorithm is shown to be NP-hard, but optimal solutions can still be obtained for realistic problem instances in reasonable times. 
Another proposed technique \cite{Asghari2016,Asghari2017} involves a truthful DRS mechanism based on a second-price auction with reserve prices.

The majority of studies on ride-sharing auctions (Table \ref{tab:research}) use two-dimensional models that perform one-to-one assignments between buyers and sellers. 
Nonetheless, DRS outputs inherently consist of one-to-many assignments for trips that contain at least three participants (one driver, two riders), whose trip-time utilities are interdependent. This limitation was partially addressed by previous studies \cite{lam2016combinatorial, Yu2018} which, however, did not consider detour effects. 
An alternative approach  \cite{Asghari2017} utilised sequential rider-vehicle matches, but without accounting for the effect of detours on valuations, assignments and pricing estimates.

To address this literature gap, we develop a mathematical model that considers the effects of shared-ride detours through a winner determination process.
This implements a sealed-bid CDA, with simultaneous driver-rider assignments that seek to maximise the total trade surplus. 
To reduce the problem search space, we build upon the concept of shareability networks \cite{santi2014quantifying}, and transform the formulation into a Maximum Weighted Independent Set Problem (MWIS), which is known to be APX-hard. 

Our contribution is summarised as follows:
\begin{enumerate}
    \item We propose a WDP model for DRS assignment, implementing a CDA while considering the effect of detours on the valuations of auction participants.
    \item We provide a local search algorithm which produces approximate results in polynomial time using greedy heuristic solutions as initializers.
    \item We identify the effects of shill bidding on our proposed CDA and suggest a robust trip price determination methodology.
\end{enumerate}

The paper is structured as follows:
Section \ref{sec:3} outlines our proposed assignment and pricing methodology for shared rides.
An exact implementation of the model and an approximate heuristic are described in section \ref{sec:4} alongside a case study for a hypothetical TNC in New York.
Findings and recommendations for future work are provided in section \ref{sec:5}.

\begin{table}[]
\small
\caption{Auction studies on ride-sharing assignment}
\label{tab:research}
\begin{tabular}{@{}lllll@{}}
\toprule
Study & Problem & Auction & Assignment & Detours \\ \midrule
\cite{Zhang2016} & DRS & CDA & one-to-one & No \\
\cite{Kleiner2011} & DRS & Vickrey & one-to-one & Yes \\
\cite{Zhao2014} & DRS & CDA & one-to-one & Yes \\
\cite{lam2016combinatorial} & DARP & VCG & \begin{tabular}[c]{@{}l@{}} one-to-many\end{tabular} & No \\
\cite{Yu2018} & DARP & CDA & \begin{tabular}[c]{@{}l@{}}one-to-many\end{tabular} & No \\
\cite{Asghari2017}, \cite{Asghari2016} & DARP & \begin{tabular}[c]{@{}l@{}}VCG\end{tabular} & \begin{tabular}[c]{@{}l@{}}one-to-one \end{tabular} & Yes \\ \bottomrule
\end{tabular}
\end{table}

\section{Methodology} \label{sec:3}
Our model assumes that travellers request shared rides through a central TNC platform that operates its own vehicle fleet.
Alongside origin/destination coordinates, travellers also submit their trip valuations. 
Vehicles have a fixed per-minute cost rate that is known in advance by the platform. 
The objective of the model is to maximise the trade surplus, defined as the sum of differences between traveller valuations and vehicle costs.

Assignments are performed in intervals with duration $\Delta$ - given the larger pool of possible matches; this quasi-online approach is expected to outperform a possible first-in-first-out (FIFO) alternative (\cite{santi2014quantifying, Yu2018}).
Two assignment types are considered: the first is between riders willing to share a trip (i.e. rider-rider), and the latter pertains to vehicles that would like to offer trips (vehicle-riders). 
In both cases, the algorithm seeks to identify potentially combinable requests, therefore establishing shareability networks \cite{santi2014quantifying} that serve as inputs to the CDA model alongside rider trip valuations. 
Any vehicles or travellers that are not matched by the CDA are deferred to later model executions alongside any requests that might have emerged in the meantime.

\subsection{Pre-matching} \label{sec:3.1}
The pre-matching stage is used to filter incompatible\footnote{Incompatible combinations produce large wait and/or detour times for riders in the combination.} vehicle-rider and rider-rider combinations before the execution of the CDA, therefore reducing instance sizes without penalising solution quality.
Quality indices $\delta_w$ and $\delta_d$ are used to reflect the maximum allowable rider wait time, and detour\footnote{Detour is defined as the additional in-vehicle time of a shared trip from a private trip that a rider might experience.} respectively. 
Let $R$ represent a set of ride requests and $K$ a set of vehicles operated by the platform. 

For each vehicle $k\in K$ we seek to obtain a subset $N_k \subseteq R$ that the vehicle can access within a period with approximate duration $\delta_w$. 
Conversely, for each ride request $r \in R$, we seek to identify a subset $A_r \subseteq K$ that can be picked up within $\delta_w$. 

A ride request $r$ is placed in $N_k$ and a vehicle $k$ is placed in $A_r$ according to Algorithm \ref{vehicle2rider} if condition $C_0$ (eq. \eqref{eq:1}) is met, where $T(\langle k,r\rangle)$ is the travel time from the current location of vehicle $k$ to the origin of request $r$, and $T(c)$ is the execution time of a stop sequence $c$.

\begin{align}
    C_0: \quad T(\langle k,r\rangle) \leq \delta_w \label{eq:1}
\end{align}

\begin{algorithm}
\caption{Prematching check: Vehicle-Rider} \label{vehicle2rider}
\begin{algorithmic}
\For{$k \in K$}
\For{$r \in R$}
\If{$C_0$}
\State $N_k \gets N_k \cup r$
\State $A_r \gets A_r \cup k$
\EndIf
\EndFor
\EndFor
\end{algorithmic}
\end{algorithm}

In the case of rider-rider matching,  we obtain the subset of second requests $I_{r} \subseteq R \setminus r$ that can be matched with a request $r \in R$ and executed with a detour lasting $\delta_d$ or less. 
We also obtain a subset of requests $J_r \subseteq R \setminus r$ that can be matched with $r$ as the second rider in the vehicle, also with a detour of $\delta_d$ or less. 
As such, for every request pair $i,j \in R, i \neq j$ where $i$ and $j$ are the first and second rider, respectively, there exists a set of origin-destination combinations $\langle o_i, o_j, d_i, d_j\rangle$ and $\langle o_i, o_j, d_j, d_i\rangle$. The following conditions apply:

\begin{align}
C_1: & \quad T(\langle o_i, o_j, d_i\rangle) \leq P_i + \delta_d  \label{eq:2}\\
C_2: & \quad T(\langle o_i, o_j, d_i, d_j\rangle) \leq P_j + \delta_d \label{eq:3}\\
C_3: & \quad T(\langle o_i, o_j, d_j, d_i\rangle) \leq P_i + \delta_d \label{eq:4}\\
C_4: & \quad T(\langle o_i, o_j, d_j\rangle) \leq P_j + \delta_d \label{eq:5}
\end{align}

In eq. \eqref{eq:2}-\eqref{eq:5}, $P_r$ represents the travel time for a private trip $r \in R$. Algorithm \ref{rider2rider} is used to prematch rider pairs - since these are obtained alongside vehicle-rider pairs the complexity of these operations relates to the cardinality\footnote{$|S|$ denotes the cardinality of any set $S$ in this paper.} of set $R$ and is  $O(|R|^2)$ \cite{santi2014quantifying}. The maximum possible total detour and waiting time for any rider $r \in R$ once the assignment is confirmed is $\delta_w + \delta_d$ due to pre-matching.

\begin{algorithm}
\caption{Prematching check: Rider-Rider} \label{rider2rider}
\begin{algorithmic}
\For{$i \in R$}
\For{$j \in R \setminus i$}
\If{$(C_1 \wedge C_2) \vee (C_3 \wedge C_4)$}
\State $I_i \gets I_i \cup j$
\State $J_j \gets J_j \cup i$
\EndIf
\EndFor
\EndFor
\end{algorithmic}
\end{algorithm}

The resulting adjacency subsets $N_k$, $A_r$, $I_r$ and $J_r$ can be visualised using a network where nodes represent vehicles or ride requests. 
A link from a vehicle $k$ to rider $r$ exists if $r \in N_k$ (and consequently $k \in A_r$), whereas a link between riders $i$ and $j$ exists if $j \in I_i$ (and consequently $i \in J_j$) or vice-versa. 
Figures \ref{fig:1a} and \ref{fig:1b} illustrate the auction participants' initial locations and the result of pre-matching respectively, in a randomly generated problem instance of 20 vehicles and 40 riders in Manhattan, New York City. 

\begin{figure}[]
\centering
\mbox{
\subfloat[]{\includegraphics[width=1.53in]{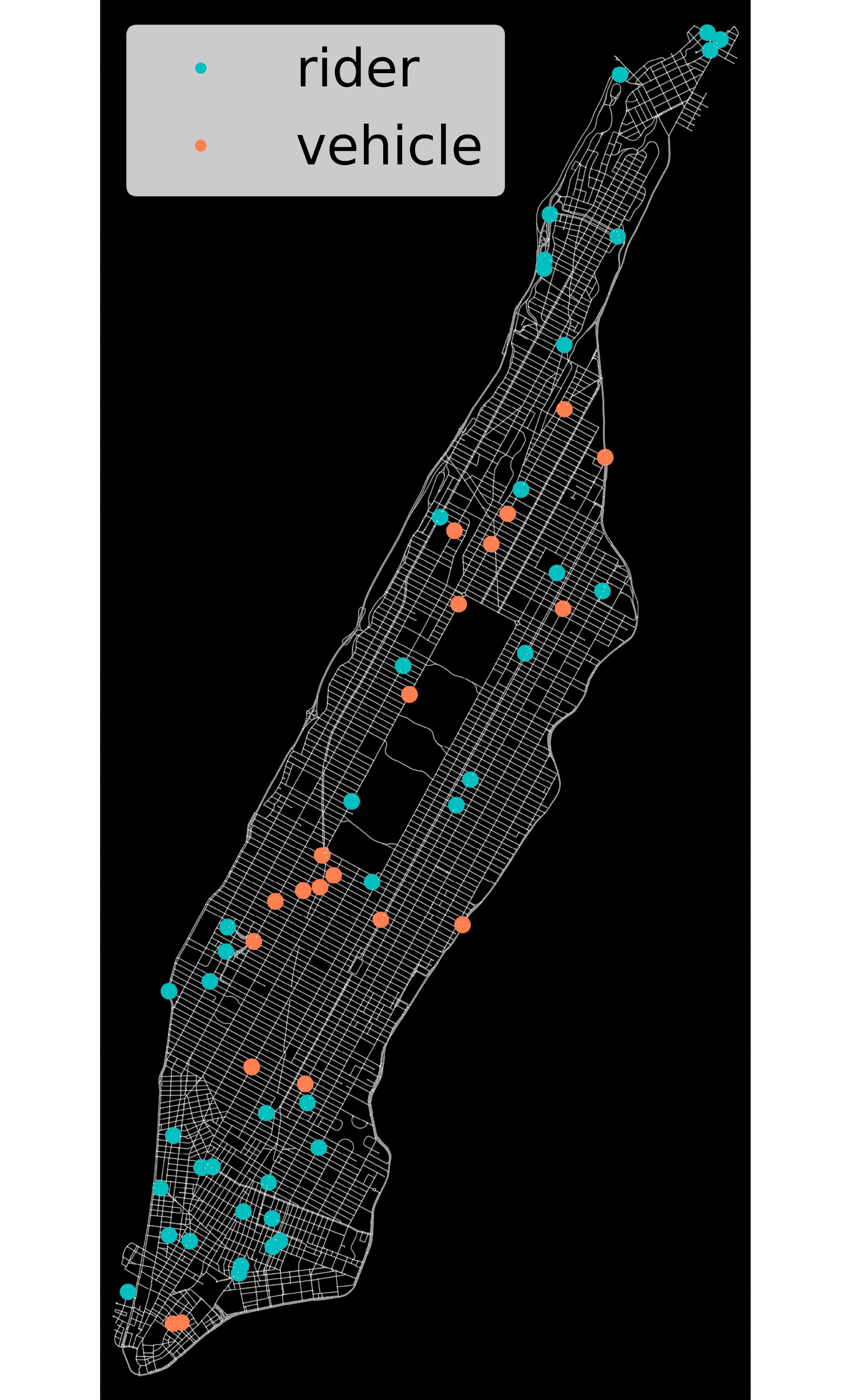}\label{fig:1a}}
\subfloat[]{\includegraphics[width=1.53in]{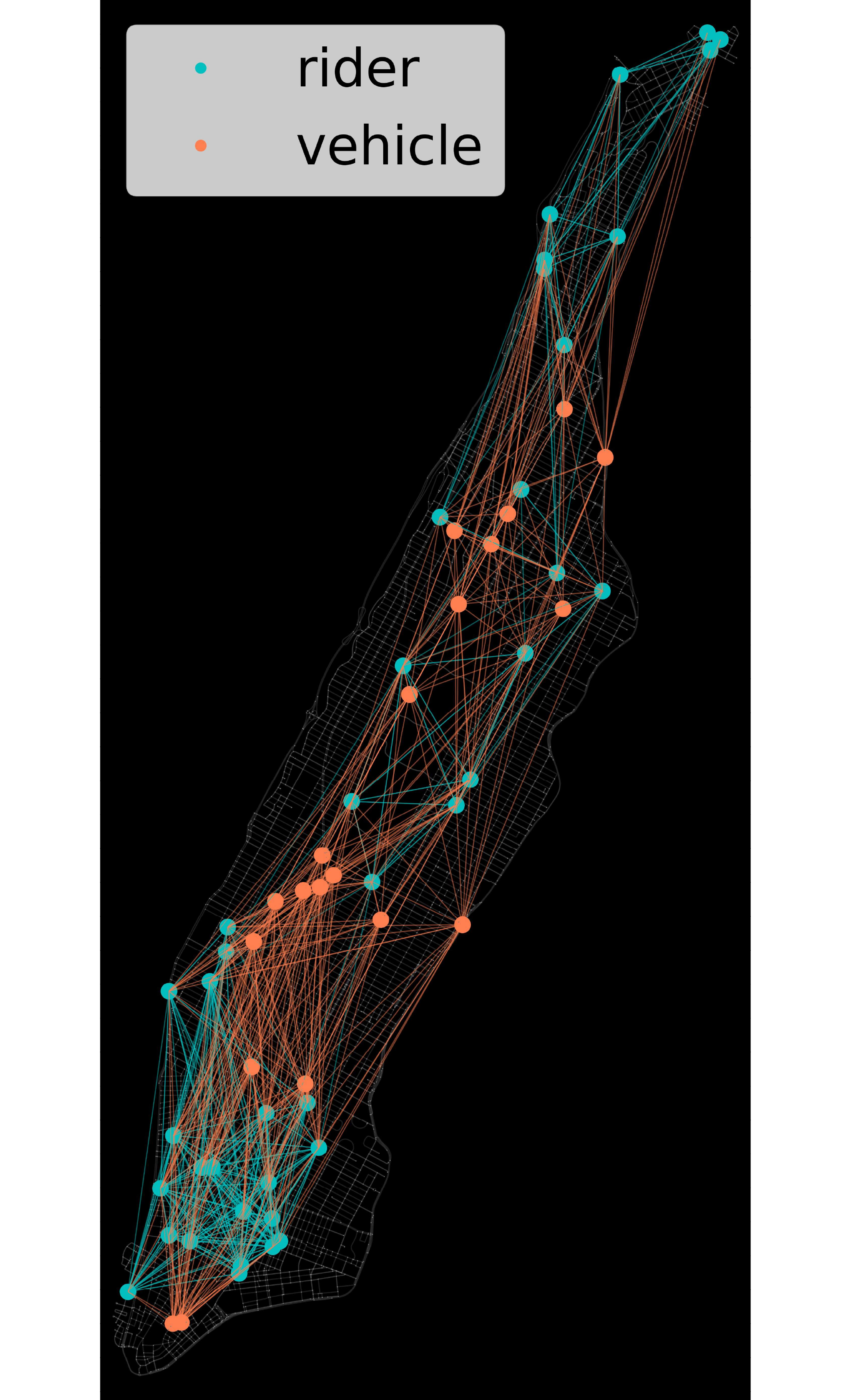} \label{fig:1b}}
}

\caption{Problem instance before (a) and after (b) pre-matching.} \label{fig:1}
\end{figure}

\subsection{Combinatorial Double Auction} \label{sec:3.2}
Our auction model builds upon \cite{Yu2018} by introducing a trading good and applying a shareability network to reduce search space.
Furthermore, it takes into account the quality of shared trips and the proximity of vehicles to achieve higher time savings.
As a result, riders would obtain different overall trip valuations when matched to different passengers or vehicles, while the pool of potential assignments would be further honed due to the use of a trip compatibility network. 
Without loss of generality, we assume that individual trip requests only consist of single riders. 
This assumption can be relaxed to extend the model to cater for larger passenger groups. 

We consider a set of riders $R$ and a set of vehicles $K$. 
Each rider $r \in R$ is identified as a 6-element tuple $\langle F_r, C_r, P_r, I_r, J_r, A_r \rangle$, where 
$F_r$ is the maximum reservation price,
$C_r$ is the time valuation,
$P_r$ is an 1D array of vehicle travel times required for a private trip (pick-up to drop-off),
while $I_r$, $J_r$ and $A_r$ are the adjacency subsets obtained through pre-matching (Section \ref{sec:3.1}).  

The 3D array $S_{i,j,n}$ represents the remaining vehicle travel time for matched riders $i$ and $j$, once the final passenger is picked up, with the pick-up sequence in the order $\langle o_i,o_j \rangle$. 
We use 3 dimensions for $S_{i,j,n}$ to account for $i$ and $j$ having different remaining travel times once $j$ is picked up. 
For example, if $i$ is dropped off first, the remaining time for $i$ might be $T(\langle o_j, d_i\rangle)$, whereas the remaining time for $j$ could be $T(\langle o_j,d_i, d_j\rangle)$. 
At the same time, the remaining travel time for the vehicle would be $T(\langle o_j,d_i, d_j\rangle)$. 
Using the procedure described in Algorithm \ref{rider2rider}, we obtain the assignment with the shortest total vehicle time. 
Finally, the index $n$ can take values between $[1,2,3]$, denoting whether $S_{i,j,n}$ refers to the first or final passenger to be picked up, or the vehicle itself, respectively. 

The array $W_{i,j}$ is used to represent the vehicle travel time from initial vehicle locations or rider origins $i$ to rider origins $j$, for $i \in {K \cup R}$, $j \in R$. 
The binary decision variable $x_{i,j} \in \{0,1\}$ is used to indicate if a vehicle or request $i$ is assigned by the action to request $j$, such that $i \in {K \cup R}$, $j \in R$. Then let: 

\begin{equation}
T_{r, 1} = \sum_{i \in I_r}{\big[x_{r+|K|,i}\big(W_{r+|K|,i}+S_{r,i,1}\big)\big]} \label{eq:driving_time_passenger1}
\end{equation}

\begin{equation}
T_{r, 2} = \sum_{i \in J_r}{\big[x_{i+|K|,r}\big(W_{i+|K|,r}+S_{i,r,2}\big)\big]} \label{eq:driving_time_passenger2}
\end{equation}

\noindent
denote the driving times from the pick-up location of the first passenger to drop-off location of the first and second passenger respectively. Similarly, let: 

\begin{equation}
T_{r,3} = \sum_{i \in I_r}\big[x_{r+|K|,i}\big(W_{r+|K|,i}+S_{r,i,3}\big)\big] \label{eq:driving_time_both_passengers}
\end{equation}

\noindent
be the driving time from the pick-up location of the first passenger to the drop-off location of the last passenger. The total service time $t_r$ of each request $r$ is therefore defined as follows\footnote{The wait time from the initial vehicle location to first passenger pickup which the second passenger experiences is omitted for complexity reasons.}:

\begin{equation}
     t_r = \sum_{k \in A_r}{\Big[x_{k,r} W_{k,r}} + x_{k,r} T_{r,1} \Big] + T_{r,2} \label{eq:9}
\end{equation}

Using the waiting and travel time from \eqref{eq:9} we can define the reservation price $f(r)$ for rider $r$ as follows:

\begin{equation}
    f(r) = F_r - C_rt_r \label{eq:10}
\end{equation}

The utility $u_r$ of a rider with respect to request $r$ is:

\begin{equation}
    u_r = 
    \begin{cases}
        f(r) - \sum_{k \in K} p_{k,r}(t_r)& \text{if $r$ can be served}\\
        0 & \text{otherwise}
    \end{cases} \label{eq:11}
\end{equation}

\noindent
where $p_{k,r}(t_r)$ in \eqref{eq:11} is the corresponding service charge for rider $r$ when is assigned to vehicle $k$, as a function of the travel time $t_r$. 
Its value is determined by the platform and is equal to zero if vehicle $k$ is not assigned to request $r$. 
Each available vehicle $k \in K$, is described as a 3-tuple $\langle B_k, Q_k, N_k \rangle$; where $B_k$ is its marginal operational cost, $Q_k$ is its capacity before assignment and $N_k$ is a subset defining riders in its vicinity (calculated as per Section \ref{sec:3.1}). 
We define the travel time $d_k$ to serve a particular set of riders for vehicle $k$, from starting to travel to the first rider until the delivery of the last rider as follows:

\begin{equation}
    d_k = \sum_{r \in N_k}\Big[x_{k,r}W_{k,r} + x_{k,r} T_{r,3}\Big] \label{eq:12}
\end{equation}

Using eq. \eqref{eq:12}, we define the cost of serving the riders assigned to each vehicle $k$ as:

\begin{equation}
    b(k)=B_kd_k \label{eq:13}
\end{equation}

As such, the total utility for vehicle $k$ when included in the auction process is defined by:

\begin{equation}
    \mu_k = 
    \begin{cases}
        \sum_{r \in R} p_{k,r}(t_r) - b(k),& \text{if $k$ serves any ride}\\
        0, & \text{otherwise.}
    \end{cases} \label{eq:14}
\end{equation}

To identify the winners of the auction and the assignment of vehicles to riders, we adopt a WDP methodology that simultaneously considers all rider bids and vehicle costs.
To achieve this, we modify the structure of the existing formulation to ensure that utilities equal to zero if rider $r$ cannot be served or vehicle $k$ is not assigned, for rider and vehicle utilities respectively.

Since $t_r$ and $d_k$ both equal to zero if rider $r$ or vehicle $k$ are not included in any assignments, the versions of the rider utility $u_r$ and vehicle utility $\mu_k$ are transformed as follows:

\begin{align}
    u_r & = X_r F_r - C_rt_r - \sum_{k \in K} p_{k,r}(t_r) & \label{eq:15} \\
    \mu_k & = \sum_{r \in R} p_{k,r}(t_r) - b(k) & \label{eq:16}
\end{align}

\noindent
where the term $X_r = \Big(\sum_{k \in A_r}x_{k,r}+\sum_{i \in I_r}x_{i+|K|,r}\Big)$ indicates whether rider $r$ is in the auction either as a first or as a second client.
The model aims to maximise the total utility of all the participants (vehicles and riders), with the objective function defined as follows:

\begin{equation}
\label{eq:17}
\begin{aligned}
    & SW  =  \sum_{r \in R} u_r + \sum_{k \in K} \mu_k & \\
    & = \sum_{r \in R} \Big(X_r F_r - C_rt_r \Big) - \sum_{k \in K}b(k) &
\end{aligned}
\end{equation}

\noindent
where $SW$ indicates the value of social welfare. Observe that the service charges cancel out in the summation of the participants' utilities. The optimisation problem is then formulated with the following set of constraints:

\textit{Model 1 (Winner Determination Problem for Ride Sharing):}

\begin{maxi!}[3]{}{SW}{}{}{}
    \addConstraint{x_{k,r} + \sum_{i \in I_r}x_{r+|K|,i}}{\leq Q_k,\quad \forall k \in K, \forall r \in N_k}{}{\label{(18b)}}
    \addConstraint{\sum_{k \in A_r} x_{k,r}+\sum_{i \in J_r}x_{i+|K|,r}}{\leq 1,\quad \forall r \in R}{}{\label{(18c)}}
    \addConstraint{\sum_{k \in A_r}x_{k,r}-1}{\leq M \Big(1-\sum_{i \in I_r}x_{r+|K|,i}\Big),\forall r \in R}{}{\label{(18d)}}
    \addConstraint{1-\sum_{k \in A_r}x_{k,r}}{\leq M \Big(1-\sum_{i \in I_r}x_{r+|K|,i}\Big),\forall r \in R}{}{\label{(18e)}}
    \addConstraint{\sum_{r \in R} x_{i,r}}{\leq 1, \quad \forall i \in K \cup R}{}{\label{18f}}
    \addConstraint{\sum_{i \in K \cup R} x_{i,r}}{\leq 1, \quad \forall r \in R}{}{\label{18g}}
    \addConstraint{\sum_{k \in A_r} x_{k,r} - \sum_{i \in I_r} x_{r+|K|,i}}{= 0, \quad \forall r \in R}{}{\label{18h}}
    \addConstraint{x_{i,j} \in \{0,1\}, \quad\forall i,j \in K \cup R}{}{}{\label{(18i)}}
\end{maxi!}

Eq. \eqref{(18b)} ensures that the number of assigned riders to each vehicle $k$ is at most equal to the vehicle capacity $Q_k$ if assigned with a rider $r$. 
\eqref{(18c)} guarantees that if rider $r$ is assigned, it is either the first rider or the second passenger to board. 
Eqs. \eqref{(18d)} and \eqref{(18e)} utilize the Big $M$ method \cite{griva2009linear} to ensure that if any two riders are matched, the first rider $r$ in the matching has to be picked up by a vehicle $k$. 
$M$ is defined as a sufficiently large positive number. 

Eq. \eqref{18f} ensures that each vehicle or rider is assigned as a starting point towards a rider at most once. 
Eq. \eqref{18g} ensures that each rider is assigned as a destination from a vehicle location or a rider no more than once. 
Finally, eq. \eqref{18h} ensures that if a vehicle is connected to a rider, there would be an additional rider in the trip.

Note that eqs. \eqref{eq:9} and \eqref{eq:12} feeding into the objective function, include non-linear terms.
We therefore introduce variables $y_{k,r} \in \mathbb{R}^+$ and $z_{k,r} \in \mathbb{R}^+$, to replace the non-linear terms in equations \eqref{eq:9} and \eqref{eq:12} respectively as shown in equations \eqref{eq:19} and \eqref{eq:20}.

\begin{align}
    t_r & = \sum_{k \in A_r}{\Big(x_{k,r} W_{k,r}} + y_{k,r}\Big)  + T_{r, 2} &   \label{eq:19}\\
    d_k & = \sum_{r \in N_k}\big(x_{k,r}W_{k,r} + z_{k,r}\big) &  \label{eq:20}
\end{align}

Consequently the objective function in equation \eqref{eq:17} transforms into the following: 

\begin{equation}
\label{eq:21}
\begin{aligned}
 & SW_L = \sum_{r \in R} u_r + \sum_{k \in K} \mu_k &\\
    & = \sum_{r \in R} \Big( X_r F_r - C_rt_r\Big) - \sum_{k \in K} B_k d_k & \\
     & = \sum_{r \in R} \Bigg[ X_r F_r - C_r \bigg[\sum_{k \in A_r}{\Big(x_{k,r} W_{k,r}} + y_{k,r}\Big)
    + T_{r,2} \bigg] \Bigg]\\
    & - \sum_{k \in K}B_k \bigg[\sum_{r \in N_k}\Big(x_{k,r}W_{k,r} +z_{k,r}\Big)\bigg] &
\end{aligned}
\end{equation}

\noindent
where $SW_L$ denotes the value of the objective after linearization. To ensure that the variable $y_{k,r}$ equals its desired value, we introduce the following linearization constraints:

\begin{align}
    y_{k,r} & \leq M x_{k,r} & \label{eq:22}\\
    y_{k,r} & \leq T_{r,1} & \label{eq:23}\\
    y_{k,r} & \geq T_{r,1} - M(1-x_{k,r}) & \label{eq:24}\\
    y_{k,r} &\in \mathbb{R}^+ & \label{eq:25}
\end{align}

\noindent
for every $r \in R$ and every $k \in A_r$. In a similar fashion, we introduce the following linearization constraints for variable $z_{k,r}$:

\begin{align}
    z_{k,r} &\leq M x_{k,r} & \label{eq:26}\\
    z_{k,r} &\leq T_{r,3} &\label{eq:27}\\
    z_{k,r} &\geq T_{r,3} - M(1-x_{k,r}) & \label{eq:28}\\
    z_{k,r} &\in \mathbb{R}^+ & \label{eq:29}
\end{align}

\noindent
for every $k \in K$ and every $r \in N_k$.

By incorporating the additional variables and constraints in equations \eqref{eq:19}-\eqref{eq:29}, our optimisation methodology for \textit{Model 1} transforms to the following Mixed Integer Linear Program (MILP):

\textit{Model 2 (Transformed WDP for Ride Sharing)}

\begin{maxi*}[2]{}{SW_L}{}{}{}
    \addConstraint{(18b)\mbox{ - }(18i)}{}{}
    \addConstraint{(22)\mbox{ - }(29)}{}{}
\end{maxi*}

\subsection{Reduction to Maximum Weighted Independent Set} \label{sec:3.3}

To assess the complexity of \textit{Model 2}, we present a reduction to the MWIS problem. 
We assume that in the largest instance, all vehicles can be matched to all requests, and all requests are sharing-compatible. 
In that scenario, with $\mathcal{K}$ and $\mathcal{R}$ being the sets of vehicles and requests, respectively, we let $C$ denote the set of all possible combinations, where $|\mathcal{C}| = |\mathcal{K}| |\mathcal{R}|^2 - |\mathcal{K}| |\mathcal{R}|$. 

Assuming that all vehicles will be assigned, the set of all path-vehicle allocations is $\binom{|\mathcal{K}| |\mathcal{R}|^2 - |\mathcal{K}| |\mathcal{R}|}{|\mathcal{K}|}$. 
To prove the APX-hardness of \textit{Model 2}, we use an approximation-preserving reduction from MWIS.

\begin{theorem}
\textit{Model 2} is NP-Hard
\end{theorem}
\begin{proof}
We reduce an instance of MWIS, a known APX-hard\footnote{APX is the complexity class of optimization problems that cannot be approximated within some constant factor unless $P \ne NP$}  problem \cite{papadimitriou1991optimization}, to an instance of \textit{Model 2}. 
Given a weighted graph $G=(V,E,w)$, the MWIS objective is to find a set of pairwise disjoint nodes $S \subseteq V$ with maximum total weight.

Let the tuple $(k,i,j)$ denote the ride-sharing trip of \textit{Model 2} with vehicle $k$ in which the first passenger is $i$ and the second is $j$, $\forall k \in \mathcal{K}, i, j \in \mathcal{R}$ and $i \ne j$. 
Also let $u_{i}(k,i,j)$ and $u_{j}(k,i,j)$ denote the utilities of riders $i$ and $j$ respectively, for the trip $(k,i,j)$ and $u_{k}(k,i,j)$ denote the utility of the vehicle. 

Consider now the following representation; let $G=(V,E,w)$ be a graph where each vertex represents a combination $c=(k,i,j)$.
An edge exists between vertices $c_n$ and $c_m$ if and only if the trip combinations $c_n$ and $c_m$ have a common element, i.e. a common vehicle or rider. 
Let: 

\begin{equation} \label{eq:node-weights}
w_c = u_k(k,i,j) + u_i(k,i,j) + u_j(k,i,j)
\end{equation}

\noindent
denote the weight of vertex $c=(k,i,j)$. 
We note that changing the order of two riders in a combination can result in a different weight for the corresponding vertex. That is because the detour or the wait time after the reordering can exceed either of the thresholds $\delta_d$, $\delta_w$ set during pre-matching, thus resulting in a different value of rider utilities.

We now prove the correctness of the above transformation. 
Let $OPT(I')$ denote an optimal solution to a \textit{Model 2} instance $I'$. 
For any two trip combinations $c, c'$ that either have a common rider or vehicle, at most one of them will be in $OPT(I')$ and the vertices representing these trips will be connected by an edge in graph $G$. 
As a result $OPT(I')$ is represented by a set of independent nodes in $G$ and since the solution is optimal with cost $\sum_{r \in R}u_r+\sum_{k \in K}\mu_k = \sum_{c \in V} w_c$ by equation \eqref{eq:node-weights} this corresponds to an independent set of maximum weight in $G$.

Conversely, suppose we have an optimal solution $OPT(I)$ on an instance $I$ of MWIS in $G$. 
Since $OPT(I)$ is independent, no pair of nodes will be connected, so no pair of trips from WDP will have a common element. 
Again according to eq. \eqref{eq:node-weights}, the total weight of the selected trips is maximised.
\end{proof}

We notice that the above reduction preserves the approximation \cite{papadimitriou1991optimization}. 
Let $f$ be the (polynomial time) transformation from an instance $I'$ of \textit{Model 2} to an instance $I$ of MWIS as described above i.e. $I = f(I')$ and let $g$ be the (polynomial time) algorithm that produces a solution to $I$ given a solution to $I'$. 
Let also $\alpha = 1$ and $\beta = 1$. 
Using transformation $f$, the optima of $I$ and $I'$ satisfy the following inequality $OPT(I') \leq \alpha OPT(I)$. 
Furthermore, having a solution with weight $w'$ for any instance $I'$, we can construct a solution for $I$ with weight $w$ such that $|w - OPT(I)| \leq \beta |w' - OPT(I')|$ using algorithm $g$.

\begin{corollary}
\textit{Model 2} is APX-Hard.
\end{corollary}

Many greedy approximation algorithms have been previously proposed, with their approximation ratio expressed as a polynomial in terms of the average or maximum node degree in the graph \cite{kako2005approximation}. 
We note that in the fully connected scenario, the average/maximum degree of node $c$ is $\Delta_c = |R|(|R| - 1) - 1 + (|K|-1)(4|R|-6)$. 
To demonstrate this, if we consider a combination $(k,i,j)$, there exist additional $|R|(|R| - 1) - 1$ trip combinations with vehicle $k$. 
For every other vehicle from the remaining $|K|-1$, there exist $2(|R|-1)$ trip combinations including rider $i$ and an additional $2(|R|-2)$ including rider $j$, which are not already accounted. 
Thus, simplifying $(|K|-1)(2(|R|-1) + 2(|R|-2))$ results to $(|K|-1)(4|R|-6)$.

\subsection{Local Search Algorithm using Greedy Search Initialisers} \label{sec:3.4}
We established earlier that solving the MWIS problem for a fully connected CDA scenario, would involve finding a MWIS in graphs with  $|\mathcal{C}| = |\mathcal{K}| |\mathcal{R}|^2 - |\mathcal{K}| |\mathcal{R}|$ nodes with an average/maximum node degree of $\Delta_c = |R|(|R| - 1) - 1 + (|K|-1)(4|R|-6)$.
Considering a small localised example with 10 vehicles and 20 potential riders, that would generate a network with 3800 nodes with an average/maximum node degree of 1045. 

An exact solution would, therefore, be impractical, as existing solution algorithms are slow even for a few hundreds of vertices \cite{Butenko2003}. We propose a local search algorithm based on simulated annealing (SA), a technique that has been shown to perform very well for the maximum clique problem (a similar premise, as it is the opposite of an independent set)\cite{homer1996experiments}. 

Simulated Annealing (SA) was initially proposed as a probabilistic method to solve difficult optimisation problems \cite{hwang1988simulated}. 
It aims to bring a system from an arbitrary initial state to an eventual state of minimum energy. 
Most SAs use an energy measure that is inversely proportional to the quality of the solution and is minimised using an iterative process. 
Starting from a seed solution, SA iterations generate several neighbouring solutions, which are accepted in accordance with a stochastic process. 
The process continues until the "temperature" of the problem reaches a user-defined minimum. 
A high-level structure of our SA algorithm for the MWIS problem is presented in Algorithm \ref{localsearch}.

\begin{algorithm}
\caption{SA for the Independent Set Problem} \label{localsearch}
\begin{algorithmic}
\State Generate initial solution $S_0$ for graph $G$
\State Set initial and minimum temperatures $T_0$, $T_{min}$
\State $S_{old} = S_0$
\State $E_{old} = energy(S_{old}, G)$
\State $S_{best} = S_{old}$
\State $E_{best} = E_{old}$
\State $T \gets T_0$
\While{$T > T_{min}$}
\State $S_{new} = neighbour(S_{old}, G)$
\State $E_{new} = energy(S_{new}, G)$
\If{$E_{new} < E_{best}$}
\State $S_{best} = S_{new}$
\State $E_{best} = E_{new}$
\EndIf
\State $S_{old}, E_{old} = select(S_{old}, S_{new}, E_{old}, E_{new}, T)$
\State $T = \alpha T$, \quad (where $\alpha$ is a constant and $\alpha<1$) 
\EndWhile
\State Output: $S_{best}$, $E_{best}$
\end{algorithmic}
\end{algorithm}

Algorithm \ref{localsearch} utilizes a graph $G$, constructed to identify all possible vehicle-rider-rider combinations by representing them as a set of nodes. Each node in the set is a 3-tuple, $\langle c, w_c, N_c \rangle$. 
$c$ refers to the combination of vehicle-rider-rider in the form of $\langle k,i,j \rangle$, $w_c$ refers to the weight of the node as defined in Section \ref{sec:3.3} and $N_c$ is a list of neighbouring nodes. 
It can be easily shown that the degree of each vertex is equal to $|N_c|$. 

To construct the graph we set $N_c = \emptyset$ and iterate through the network nodes to populate $N_c$ for each vertex. 
As with Algorithm \ref{vertexgen}, this process requires $|K||R|^2$ iterations (fully connected scenario) to create the set of vertices $V$. 
Populating $N_c$ for each vertex (and creating the edge set $E$), requires $|V|^2$ iterations (Algorithm \ref{edgegen}). Since $|V|^2 = (|K||R|^2)^2$, the complexity of the worst case  scenario for network generation is $O(|K|^2|R|^4)$.
This process, however, can be easily parallelised.

\begin{algorithm}
\caption{Vertex Generation Process} \label{vertexgen}
\begin{algorithmic}
\State $V \gets \emptyset$
\For{$k \in K$}
\For{$i \in N_k$}
\For{$j \in I_i$}
\State $w_c = u_{i}(k,i,j) + u_{i}(k,i,j) + u_{k}(k,i,j)$
\If{$w_c \geq 0$}
\State $c = \langle k, i, j \rangle$
\State $N_c = \emptyset$
\State $V \gets V \cup \langle c, w_c, N_c \rangle$
\EndIf
\EndFor
\EndFor
\EndFor
\State Output: $V$
\end{algorithmic}
\end{algorithm}

\begin{algorithm}
\caption{Edge Generation Process} \label{edgegen}
\begin{algorithmic}
\State Non-empty set $V$
\State $E \gets \emptyset$
\For{$i \in V$}
\For{$j \in V\setminus i$}
\If{$c_i \cap  c_j \neq \emptyset$}
\State $N_{c_i} \gets N_{c_i} \cup j$
\State $N_{c_j} \gets N_{c_j}\cup i$
\State $E \gets \langle i,j \rangle$
\EndIf
\EndFor
\EndFor
\State Output: $G=(V,E)$
\end{algorithmic}
\end{algorithm}

A set of greedy heuristics with known lower bound performance \cite{kako2005approximation} is used to obtain an initial solution $S_0$, consisting of an ordered set of vertices in $V$. 
These operate by sorting vertices in a descending order with respect to $w_c$, $1/|N_c|$, $w_c/|N_c|$ and $w_c/\sum_{i \in N_c} w_i$, respectively. 
The best solution among these four is identified through inspection. 

To calculate the energy of a solution (Algorithm \ref{energy}), we iterate through the ordered vertex sequence $S$. 
At each step, we add the next vertex in $S$ to the independent set $I$ and removing its neighbours from $S$.
Iterations continue until $S$ is empty. 
The energy of the solution is, therefore, equal to the negative sum of all values $w_c$, for each vertex within $I$. 

\begin{algorithm}
\caption{Energy Calculation} \label{energy}
\begin{algorithmic}
\State Non-empty ordered sequence $S$
\State Graph $G=(V,E)$
\State $I \gets \emptyset$
\While{$S \neq \emptyset$}
\State $i = S(1)$
\State $I \gets I \cup i$
\State $S \gets S \setminus (S \cap  (N_{c_i} \cup i))$, \quad (obtain $N_{c_i}$ from $G$)
\EndWhile
\State $E = - \sum_{i \in I} w_{c_i}$, \quad (obtain $w_{c_i}$ from $G$)
\State Output: $E$
\end{algorithmic}
\end{algorithm}

When it comes to the generation of neighbouring solutions, we randomly select two vertices in the independent set $I$ of the old solution $S_{old}$ and switch their positions in $S_{old}$ to produce sequence $S_{new}$. This approach increases the chance that sequence $S_{new}$ will produce a different independent set and energy than $S_{old}$. Finally, we form our stochastic selection method on defining an acceptance probability for every new solution, which is calculated using $E_{old}$, $E_{new}$ and temperature $T$ as shown in Algorithm \ref{selection}. Better solutions are always accepted, whereas worse solutions have less chance of being accepted as the iterations progress (i.e. as temperature $T$ decreases).

\begin{algorithm}
\caption{Selection Process} \label{selection}
\begin{algorithmic}
\State Inputs: $S_{old}$, $S_{new}$, $E_{old}$, $E_{new}$, $T$ 
\State $p = X$, \quad (where $X \sim U(0,1)$)
\If{$E_{new} < E_{old}$}
\State $p_a = 1$
\Else
\State $p_a = e^{(E_{old}-E_{new})/T}$
\EndIf
\If{$p_a > p$}
\State $S_{old} = S_{new}$
\State $E_{old} = E_{new}$
\EndIf
\State Outputs: $S_{old}$, $E_{old}$
\end{algorithmic}
\end{algorithm}

\subsection{Trip Price Determination} \label{sec:3.5}
Optimal solutions of the WDP in CDAs produce efficient outcomes which are individually rational. That is, assuming participants in the auction are truthful about their valuations. There is, however, no guarantee that auction participants (bidders) will state their true valuations. \cite{deVries} explains this problem with an example of three bidders. We will extend this example to our CDA, to illustrate how untruthful bids can arise. 

Let us consider a CDA scenario involving three riders (bidders) and one vehicle. 
Let us also assume that from the six possible allocation combinations, the following three yield a positive value for total trade surplus:

\begin{equation}\label{eq:31}
    f_1(\langle1,2\rangle)=10, \quad f_2(\langle1,2\rangle)=8, \quad b_1(\langle1,2\rangle)=10
\end{equation}

\begin{equation}\label{eq:32}
    f_1(\langle2,1\rangle)=7, \quad f_2(\langle2,1\rangle)=9, \quad b_1(\langle2,1\rangle)=11
\end{equation}

\begin{equation} \label{eq:33}
    f_1(\langle1,3\rangle)=5, \quad f_3(\langle1,3\rangle)=10, \quad b_1(\langle1,3\rangle)=12
\end{equation}

In eqs. \eqref{eq:31}-\eqref{eq:33}, $f_r(\langle S\rangle)$ and $b_k(\langle S\rangle)$ represent total valuation and cost for a rider $r$ and a vehicle $k$, respectively, for a trip with a pickup sequence $S$. 
Using \textit{Model 2}, the platform allocates the trip with the only vehicle servicing riders $1$ and $2$ in the sequence $\langle1,2\rangle$ as it is the combination producing the highest trade surplus. 
Note that riders $1$ and $2$, assuming everyone bids truthfully, can report a lower value per time and still win the auction with the same combination. 

The inclusion of additional riders will give rise to more complex bidding strategies.
In the case that riders $1$ and $2$ reduce their bids excessively, they might lose in the auction. 
This characteristic CDA property is known as the threshold problem \cite{bykowsky2000mutually} and refers to the implication of valuation misreporting thresholds for individual participants, which can motivate bidders to employ perverse bidding strategies \cite{porter2003combinatorial}.

Pricing in VCG auctions, where bidders pay the difference of welfare in their absence with the welfare of others when they are included in the auction, is incentive-compatible \cite{deVries}. 
Furthermore, incentive-compatible payments have been derived through the solution of dual relaxed linear problems (LPs) of the WDP \cite{ba2001optimal}. 
Previous studies \cite{Iosifidis2010, Xu2014, Yu2018}, used relaxed dual WDP problems to identify allocation and pricing in double auctions, with Lagrangean multipliers to be considered as prices. It has been shown that optimal dual variables in LP coincide with VCG payments \cite{bikhchandani2001linear}.

However, the use of near-optimal CDA solutions does not preserve incentive compatibility \cite{nisan2007computationally}. 
Negligible variations from the optimal objective can have significant consequences on the payments to be made by bidders \cite{johnson1997equity}. 
As such, an approximate WDP solution would inhibit the use of VCG or dual LP relaxations that would guarantee incentive-compatibility. The NP-hardness of our proposed CDA prohibits the identification of exact WDP solutions in practical implementations, thereby we omit the use of VCG or dual LP relaxations for price determination. 

Instead, we propose a model which resembles a Generalised First Price (GFP) auction for trip pricing. A GFP mechanism is an untruthful auction mechanism, where participants bid for the allocation of a limited amount of slots. Participants pay their bid values in case they are assigned to a slot. Previous research outlined deficiencies in the GFP mechanism by strategically employed shill bidding which destabilizes the auction \cite{EDELMAN2007192}. Subsequent work in \cite{doi:10.1287/moor.2017.0920} attributes these GFP deficiencies to the auction interface and argues that GFP auctions can be robust by allowing expressiveness of the participants using multidimensional bids.

In the conventional GFP, an individual $i$ submits a single bid $f_i$, which is multiplied by $s_1 \geq s_2 \geq ... \geq s_k$, $k$ being the last available slot. The expressive version of GFP dictates that an individual $i$ submits a different bid $f_{ik}$ for each slot $k$ which is multiplied by $s_1 \geq s_2 \geq ... \geq s_k$ accordingly. Our proposed CDA resembles an expressive GFP, as travellers bid for a limited number of vehicle seats (slots) and by submitting a valuation per time $C_r$, they might obtain a different valuation $f(r)$ for each potential vehicle-rider-rider assignment. 

To limit the effect of untruthful bids on the auction outcome, we propose that each rider only submits the valuation per time $C_r$. The platform in turn identifies and privately informs the rider of its maximum reservation price $F_r$, so that if matched, the payment will comprise of a discounted static price for the time of the trip attributing to $P_r$ and an additional variable rate attributing to $C_r \bar\delta_r$, where $\bar\delta_r$ is the wait and detour time saved by choosing the platform, instead of the rest of the market. 

Consequently, the maximum reservation price $F_r$ is derived by the platform using the following generalised cost equation:

\begin{equation} \label{eq:34}
    F_r = p_b + P_r p_t+ C_r (P_r + \delta_w + \delta_d)
\end{equation}

where $p_b$ is the flat fee and $p_t$ is the discounted price per minute for a shared trip, lasting $P_r$  minutes if private, as specified by the platform. $\delta_w$ and $\delta_d$ refer to the guaranteed maximum wait and detour times respectively, which are used in pre-matching by the platform.

By introducing this format, it is straight-forward to deduce by observing equations \eqref{eq:10}, \eqref{eq:13} and \eqref{eq:17} that in the event where bidders submit per time valuations $C_r$ which are very close to zero, our proposed CDA converts to an optimal 3D assignment problem where the sum of detours is minimised, if the following inequality holds for any vehicle $k$ and riders $i$, $j$ prior to the auction:

\begin{equation} \label{eq:35}
    b(k) \leq f(i) + f(j)
\end{equation}

By introducing this condition with equation \eqref{eq:35}, we ensure that the auction always returns an assignment if a pre-matching instance exists as any rider payments in the GFP instance will always cover the vehicle costs. We thereby need to choose the appropriate value for the flat fee $p_b$, such that equation \eqref{eq:35} holds. In doing so, we assume that the total rider payment per vehicle equals its cost. We also assume that vehicle costs $B_k$ are uniform across the fleet (i.e. $B_k=B \quad\forall k \in K$) and extend the functions as per equations \eqref{eq:10} and \eqref{eq:13}:

\begin{equation} \label{eq:36}
    Bd_k = F_i - C_it_i + F_j - C_jt_j
\end{equation}

Using equation \eqref{eq:34}, and by substituting $\delta_w + \delta_d$ with $\delta$, we reach to the following:

\begin{equation} \label{eq:37}
    Bd_k = 2p_b + p_t (P_i + P_j) + C_i (P_i + \delta - t_i) + C_j (P_j + \delta - t_j) 
\end{equation}

In the minimal total bid scenario, both $C_i$ and $C_j$ in equation \eqref{eq:37} would be zero. We also know that $d_k$ is equal to $max(t_i, t_j)$. By setting the total wait and detour time experienced by each rider $r$ as $\delta_r$, we can replace $t_r$ by $P_r + \delta_r$. Therefore, with $C_i$ and $C_j$ set to zero we arrive to the following equation:

\begin{equation} \label{eq:38}
    B max(P_i + \delta_i, P_j + \delta_j) = 2p_b + p_t (P_i + P_j)
\end{equation}

The maximum vehicle cost in equation \eqref{eq:38} for any values of $P_i$ and $P_j$, occurs if $max(P_i + \delta_i, P_j + \delta_j) = max(P_i, P_j) + \delta$, for $\delta$ as introduced above, being the maximum total wait and detour time guarantee by the platform for an individual rider. Assuming the value of $p_b$ is zero and that $p_t$ is set by the platform such that $p_t \geq B$, if $min(P_i, P_j) \geq \delta$ the condition in \eqref{eq:35} always holds. If however $min(P_i, P_j) < \delta$, a flat fee $p_b$ is required to ensure the condition in \eqref{eq:35}. As such, assuming both $P_i, P_j \rightarrow 0$, and  $max(\delta_i, \delta_j)=\delta$, using equation \eqref{eq:38}, the flat fee for our proposed GFP interface should be as follows:

\begin{equation} \label{eq:39}
     p_b = \frac{B \delta}{2}=\frac{B (\delta_w + \delta_d)}{2}
\end{equation}

\section{Discussion} \label{sec:4}
Our methodology was implemented using Python and tested on a workstation with an Intel i7-4790 CPU (3.6GHz) and 8GB RAM. 
Exact solutions were obtained using the Branch and Cut algorithm provided by IBM ILOG Cplex Optimization Studio 12.7.1. 

To test the algorithm, we create a case study network set in Manhattan, NY. The underline road network and travel times were obtained using the OSMnx library \cite{BOEING2017126}. 
To account for congestion, we applied a 20\% penalty to the free-flow speeds in residential and motorway link segments, and 40\% elsewhere. 
Rider origin-destination pairs, as well as vehicle locations, were sampled uniformly in space to create CDA instances. 
Only trips with travel time that is greater than 5 minutes were considered, while $\delta_w$ and $\delta_d$ were both set to 10 and 15 minutes respectively.

For this study, we used UK-based estimates of working time valuations \cite{DfT2018} for the derivation of rider valuations. Vehicles were assumed to have a capacity of two customers, with their operating costs $B_k$ uniformly set to 12.96 GBP/hour. Conversely, customer time valuations $C_r$ were sampled from a log-normal distribution with a mean of 17.69 GBP/hour and $\sigma = 0.02$. The discounted price per minute $p_t$ was set to 0.75 GBP/min.

Table \ref{tab:comparison table} provides a performance comparison of the Simulated Annealing (SA) and the Branch and Cut (BC) algorithms for a range of instances. As can be seen in the table and in figure \ref{fig:2}, the runtime for the BC approach grows exponentially as more vehicles and riders are considered in the instance, thereby increasing the node count of the MWIS instance, whereas the runtime for the SA remains relatively short. 

\begin{table}[]
\caption{Performance comparison of SA and BC.}
\label{tab:comparison table}
\begin{tabular}{lllllll}
\hline
\begin{tabular}[c]{@{}l@{}}Total\\ vehicles\end{tabular} &
  \begin{tabular}[c]{@{}l@{}}Total\\ riders\end{tabular} &
  \begin{tabular}[c]{@{}l@{}}BC\\ solution\end{tabular} &
  \begin{tabular}[c]{@{}l@{}}BC\\ runtime\\ {[}sec{]}\end{tabular} &
  \begin{tabular}[c]{@{}l@{}}SA\\ solution\end{tabular} &
  \begin{tabular}[c]{@{}l@{}}SA\\ runtime\\ {[}sec{]}\end{tabular} &
  \begin{tabular}[c]{@{}l@{}}Error\\ {[}\%{]}\end{tabular} \\ \hline
4  & 8  & 48.199  & 0       & 48.199  & 0.002 & 0.00 \\
8  & 8  & 77.714  & 0.13    & 77.714  & 0.008 & 0.00 \\
5  & 10 & 120.523 & 0.02    & 120.523 & 0.004 & 0.00 \\
10 & 10 & 127.936 & 0.13    & 127.936 & 0.027 & 0.00 \\
5  & 12 & 108.327 & 0.03    & 108.327 & 0.004 & 0.00 \\
6  & 12 & 136.996 & 0.05    & 136.996 & 0.007 & 0.00 \\
12 & 12 & 145.248 & 0.11    & 144.866 & 0.026 & 0.26 \\
6  & 14 & 122.487 & 0.05    & 122.487 & 0.021 & 0.00 \\
7  & 14 & 131.537 & 0.28    & 128.81  & 0.029 & 2.07 \\
7  & 15 & 208.895 & 0.5     & 208.895 & 0.034 & 0.00 \\
8  & 16 & 171.665 & 0.27    & 170.613 & 0.053 & 0.61 \\
7  & 17 & 160.367 & 0.06    & 160.253 & 0.024 & 0.07 \\
9  & 18 & 236.672 & 1.66    & 235.918 & 0.076 & 0.32 \\
8  & 20 & 204.96  & 0.28    & 204.93  & 0.061 & 0.01 \\
10 & 20 & 236.125 & 1.45    & 234.785 & 0.149 & 0.57 \\
11 & 22 & 295.669 & 3.45    & 291.632 & 0.209 & 1.37 \\
12 & 24 & 331.237 & 29.95   & 321.797 & 0.598 & 2.85 \\
10 & 25 & 326.057 & 39.86   & 321.967 & 0.64  & 1.25 \\
14 & 28 & 377.018 & 103.5   & 373.223 & 2.702 & 1.01 \\
15 & 30 & 397.518 & 130.92  & 387.628 & 2.374 & 2.49 \\
15 & 20 & 285.793 & 18.14   & 283.101 & 0.407 & 0.94 \\
20 & 20 & 270.432 & 58.24   & 267.002 & 0.746 & 1.27 \\
12 & 25 & 324.687 & 26.64   & 321.524 & 1.239 & 0.97 \\
15 & 25 & 323.191 & 42.3    & 316.062 & 1.399 & 2.21 \\
16 & 25 & 341.845 & 82.19   & 333.836 & 1.496 & 2.34 \\
17 & 25 & 330.884 & 124.78  & 321.94  & 1.573 & 2.70 \\
14 & 30 & 419.626 & 49.59   & 415.906 & 2.193 & 0.89 \\
13 & 26 & 329.992 & 63.22   & 324.955 & 2.587 & 1.53 \\
16 & 32 & 431.772 & 1003.55 & 408.41  & 3.423 & 5.41 \\
17 & 34 & 469.94  & 4126.59 & 445.309 & 4.456 & 5.24 \\ \hline
\end{tabular}
\end{table}

\begin{figure}[]
\centering
\includegraphics[width=0.48\textwidth]{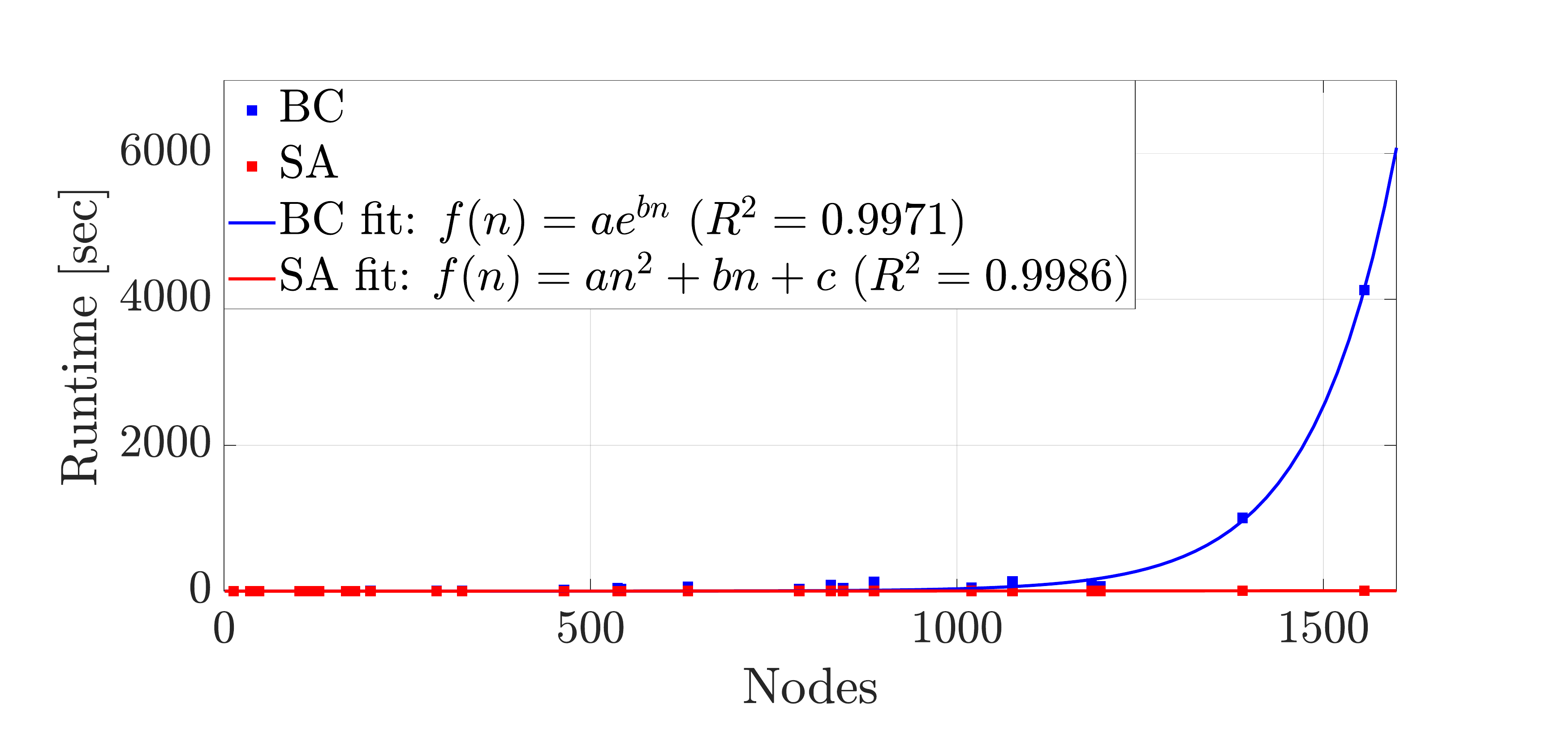}
\caption{Run-time for BC and SA methods.}
\label{fig:2}
\end{figure}

The APX-complete nature of our problem is also signified in the solution comparison between BC and SA as observed in figure \ref{fig:3}, as the percentage error gradually increases with a larger instance size. However, as shown in figure \ref{fig:4}, the approximation error is relatively low for instances of such size. A visual comparison of the results obtained by the BC and SA algorithms is provided in figures \ref{fig:5a} and \ref{fig:5b}, respectively, for an instance involving 10 vehicles, 20 customers and two edges per match outlines the similarities between solutions obtained using the two approaches.

\begin{figure}[]
\centering
\includegraphics[width=0.48\textwidth]{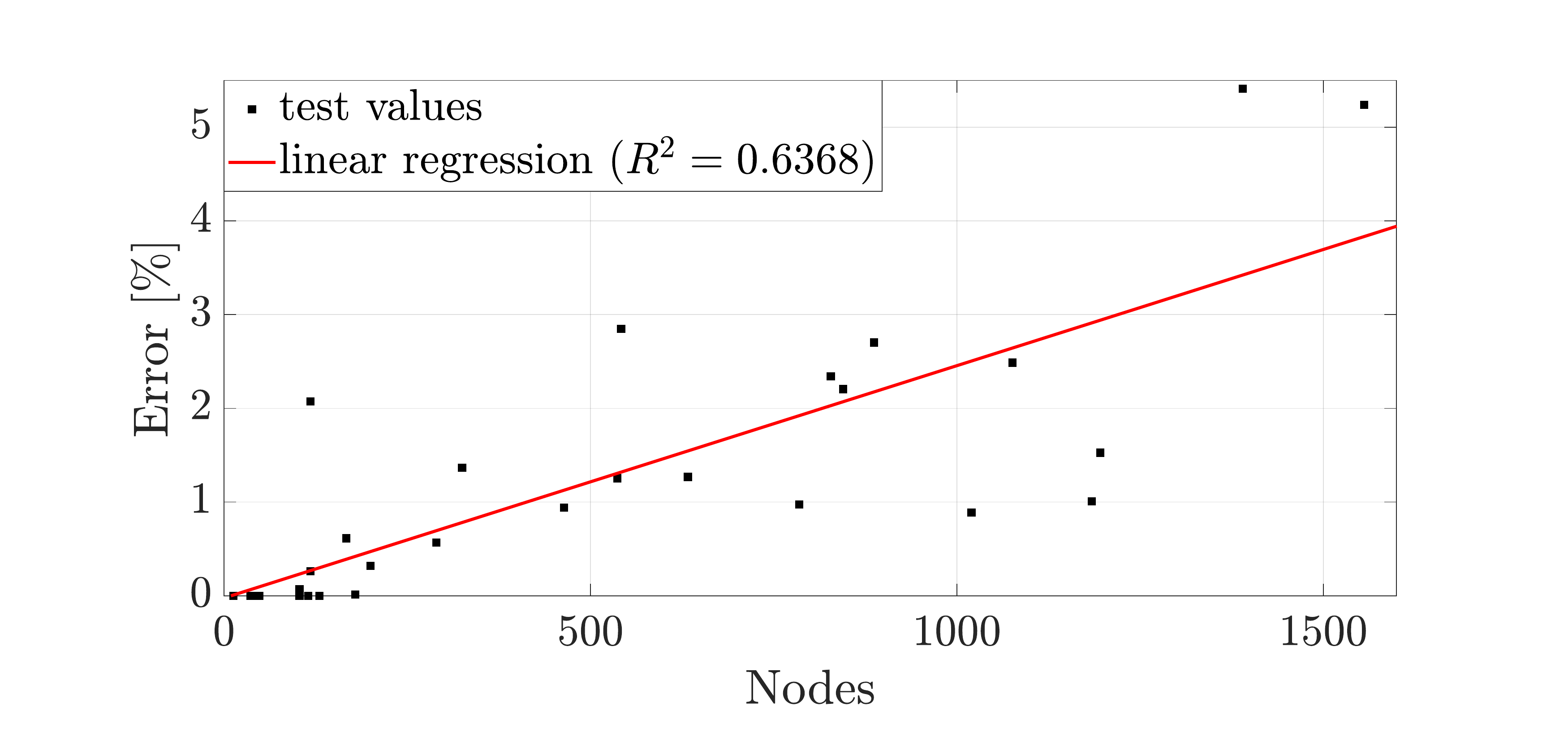}
\caption{Percentage error of approximation against node count.}
\label{fig:3}
\end{figure}

\begin{figure}[]
\centering
\includegraphics[width=0.48\textwidth]{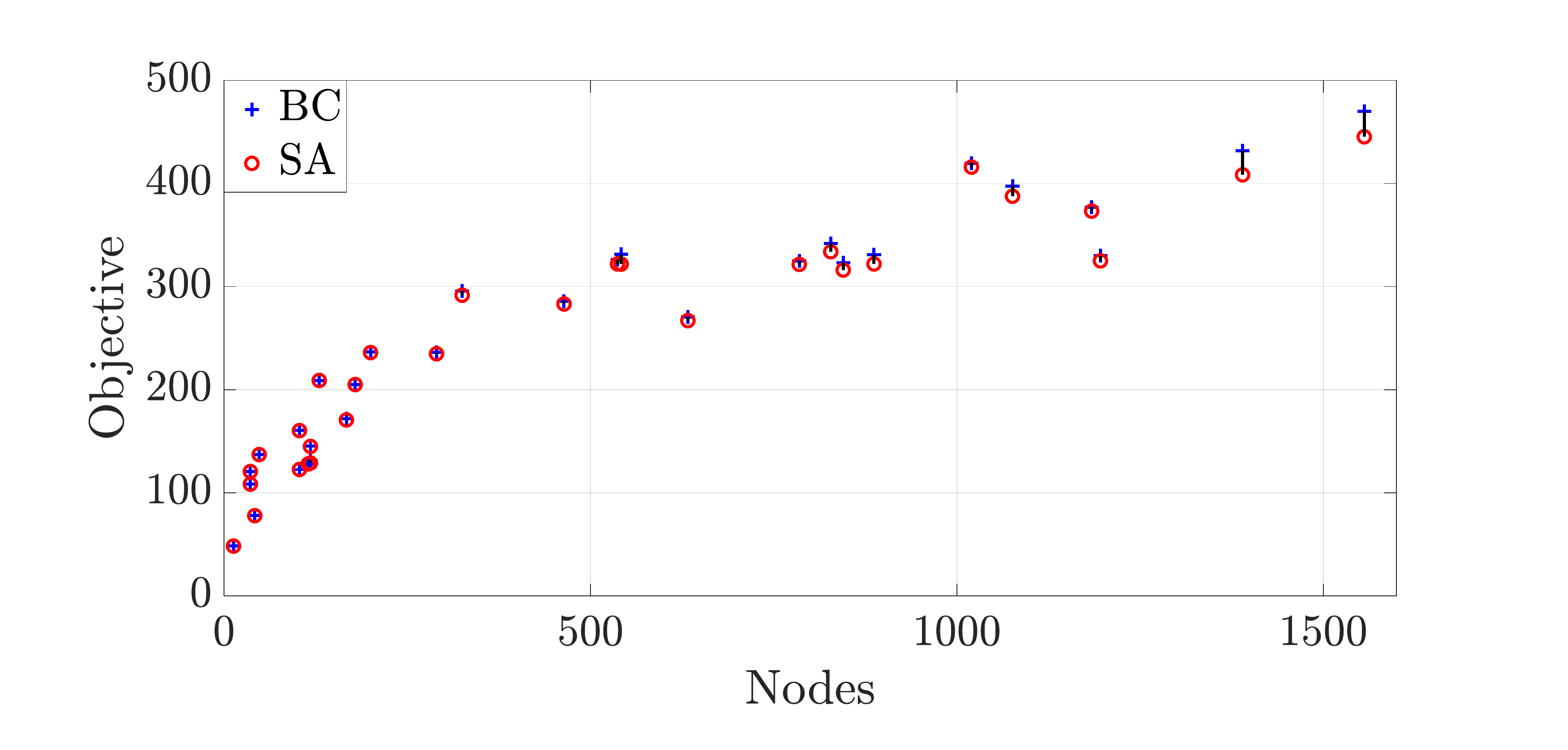}
\caption{Solution values for BC and SA methods against node count.}
\label{fig:4}
\end{figure}

\begin{figure}[H]
\centering
\mbox{
\subfloat[]{\includegraphics[width=1.53in]{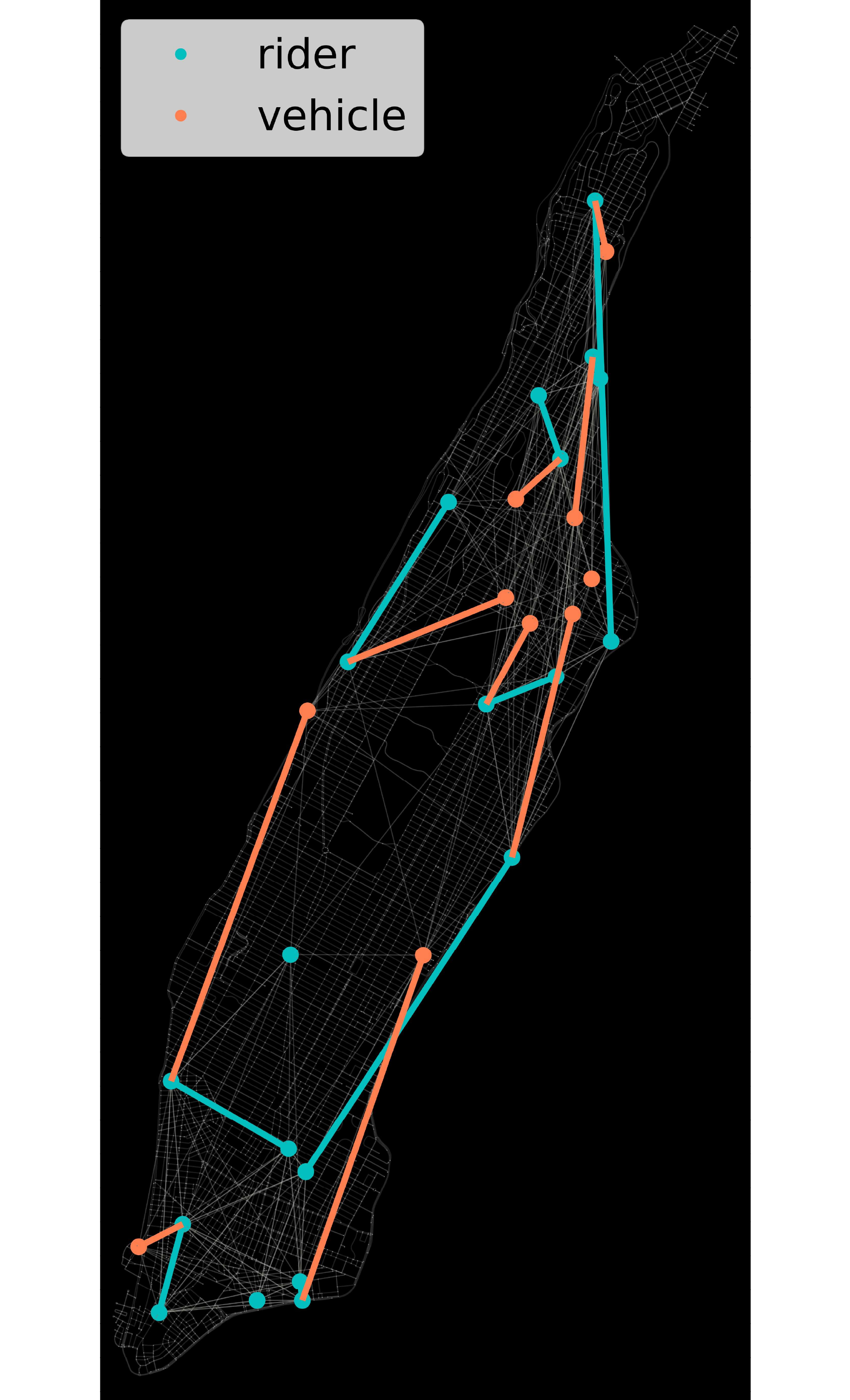}\label{fig:5a}}
\subfloat[]{\includegraphics[width=1.53in]{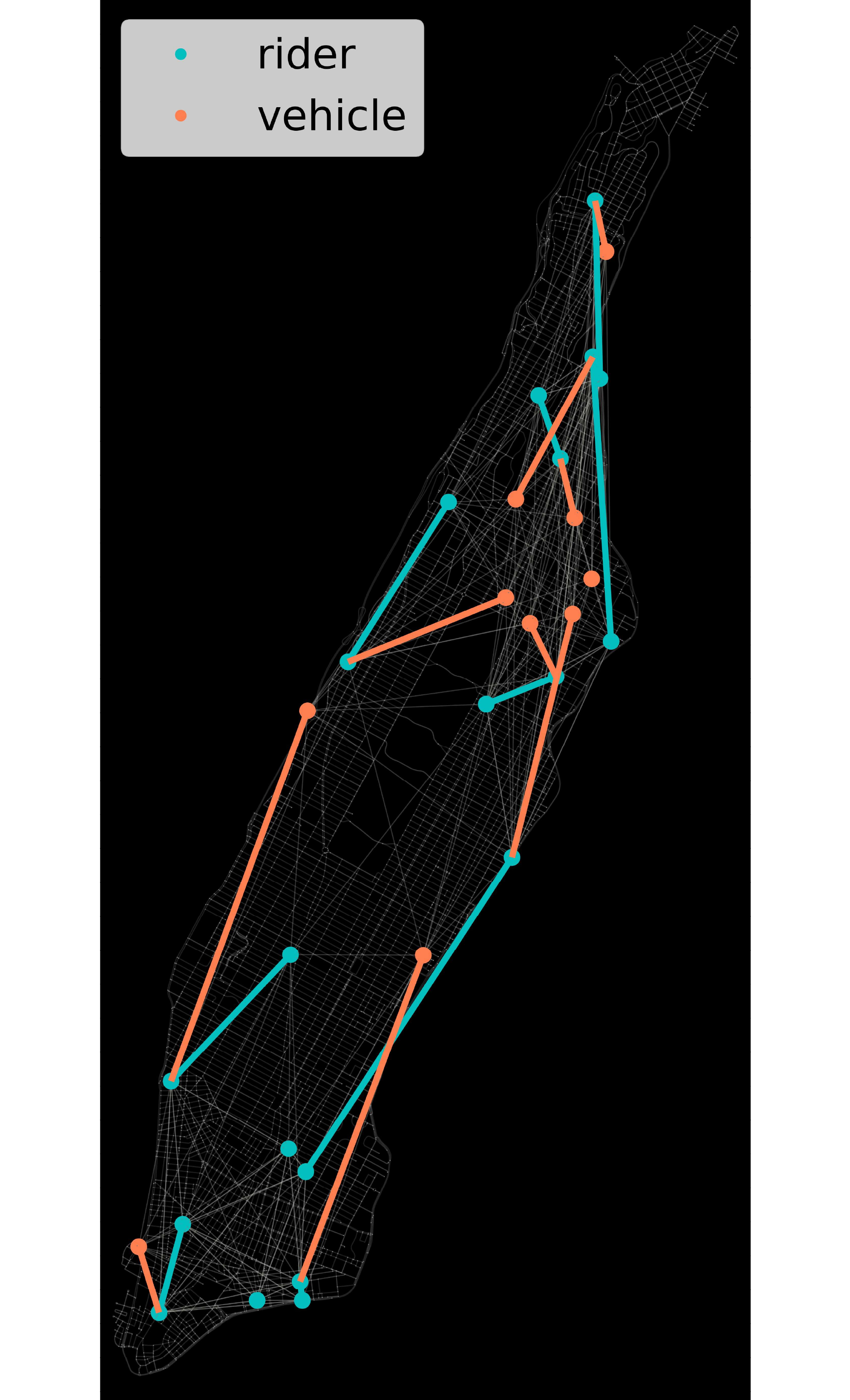} \label{fig:5b}}
}

\caption{Visualisation of (a) BC and (b) SA solutions} \label{fig:5}
\end{figure}

To strengthen the argument for the inclusion detour calculations on CDAs for ride-sharing, we conducted a comparison analysis between exact solutions of our \textit{Model 2} and the algorithm in the state-of-the-art which mostly resembles our problem statement, namely the CDA model in \cite{Yu2018}. We created instances from 10 to 22 requests, with the assumption of one seat per request. For each instance, we assumed there are just enough vehicles to cover the demand (i.e. half the number of requests). To run the CDA in \cite{Yu2018}, we converted the distance-based methodology to time-based to match \textit{Model 2} and omitted private rides. We used a vehicle capacity of two rides for all vehicles in both models.

As observed in Figure \ref{fig:6_0}, since the CDA in \cite{Yu2018} omits detours and wait times in their calculation, the resulting assignment creates much higher detour and wait times on average for each instance. Consequently, ignoring the effect of detours and wait times in ride-sharing CDAs can produce assignments which might not be acceptable by the users of the service. Taking the time dimension into account we can indeed massively improve the convenience of the service as observed. However, we achieve this with an increase in computational complexity, as discussed in \ref{sec:3}. Nonetheless, a reduction of the solution space can be achieved via the pre-matching stage, as shown in \ref{sec:3.1}.

\begin{figure}[]
\centering
\includegraphics[width=0.48\textwidth]{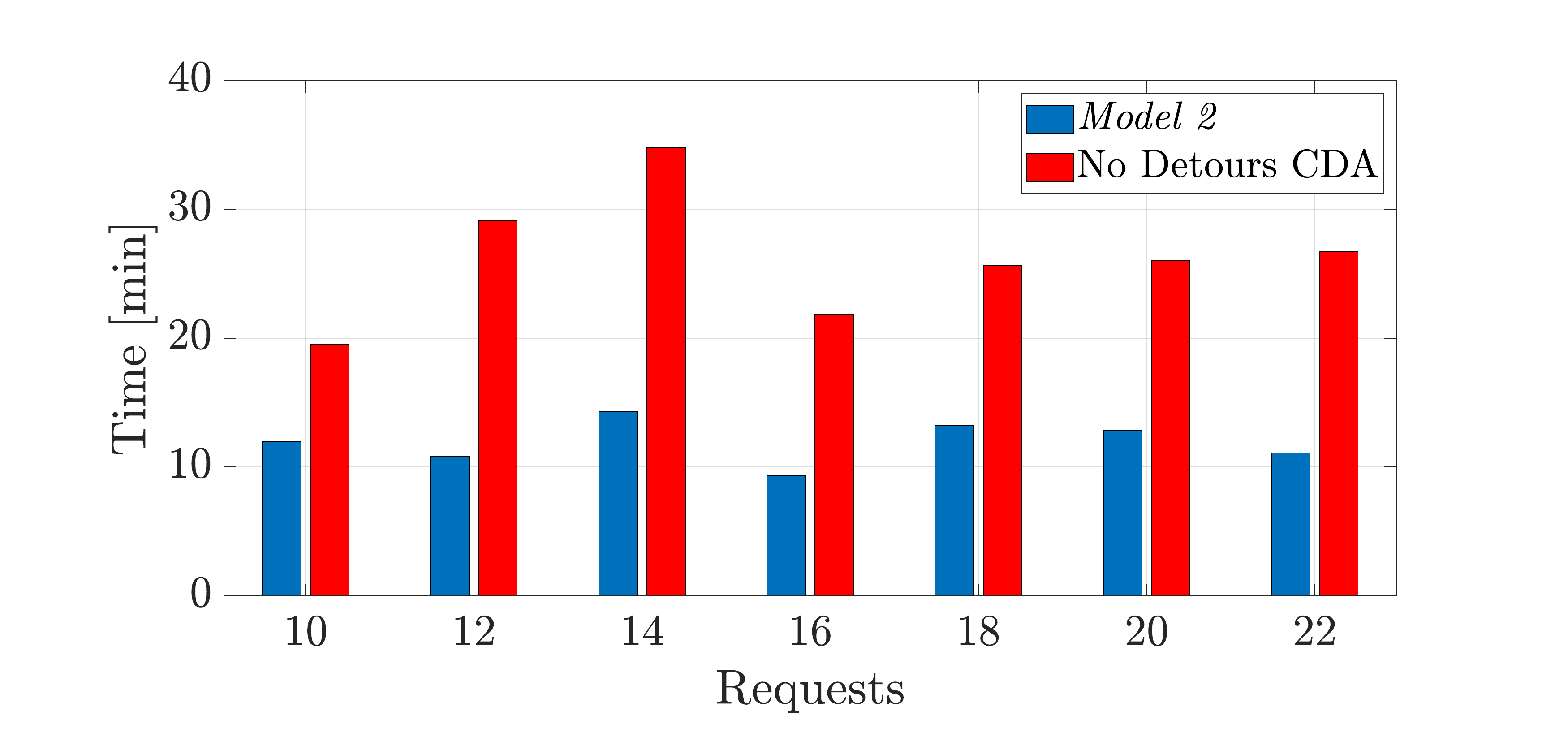}
\caption{Average total wait and detour time per request for \textit{Model 2} and CDA in \cite{Yu2018}.}
\label{fig:6_0}
\end{figure}

\subsection{Trade surplus implications}
A large number of problem instances were considered, with fleet sizes ranging between 3 and 60 vehicles, and a customer base of 10 to 60 riders. From the range of greedy heuristics that were considered for SA initialisation described in section \ref{sec:3.4}, weight-based approaches were found to yield the best results (Figure \ref{fig:6}). Figures \ref{fig:7} and \ref{fig:8} illustrate the relationship between problem sizes and algorithm run times, which is found to be in polynomial time.

\begin{figure}[h]
\centering
\includegraphics[width=0.48\textwidth]{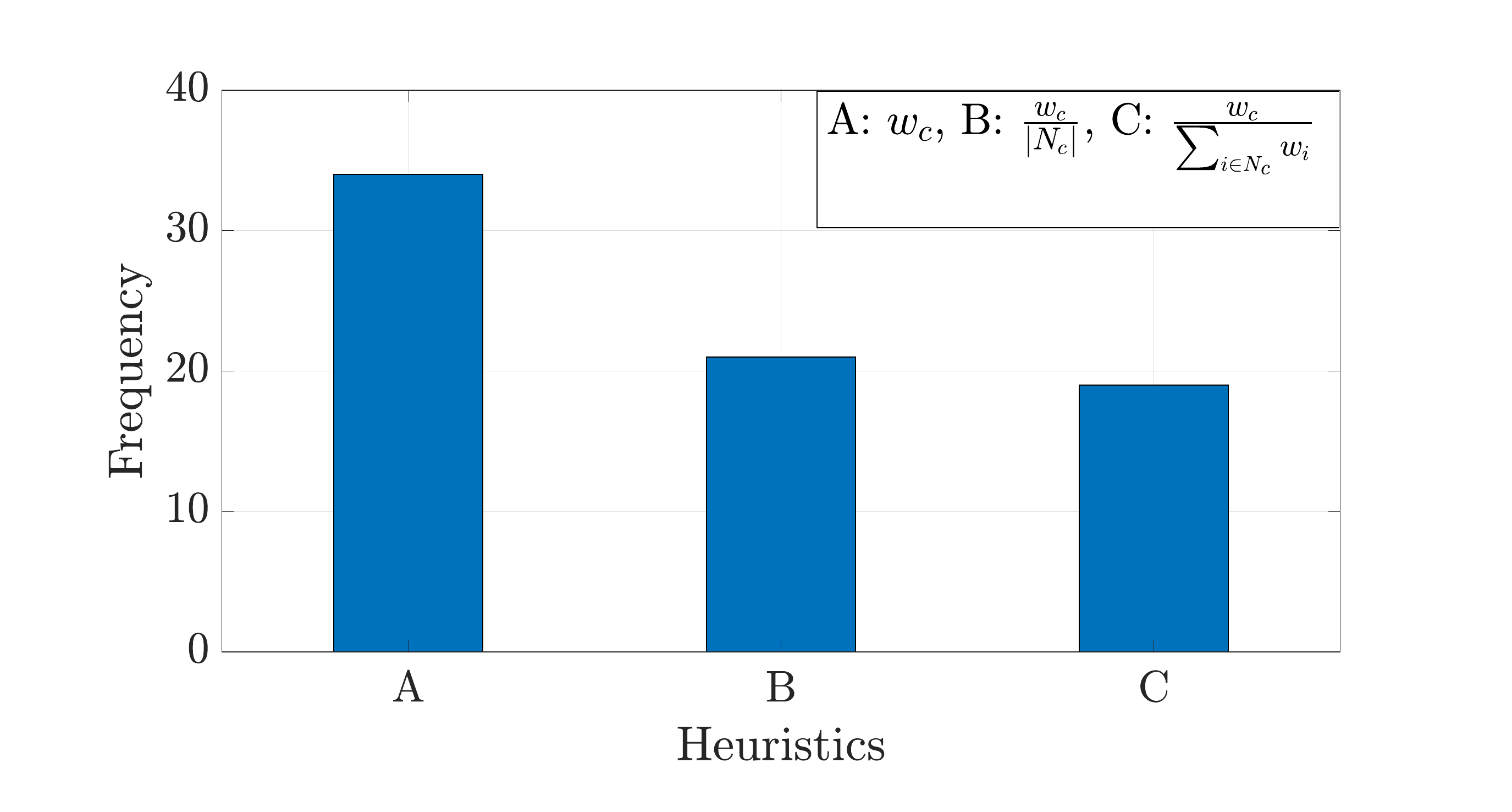}
\caption{Performance comparison of greedy initialisers}
\label{fig:6}
\end{figure}

\begin{figure}[h]
\centering
\includegraphics[width=0.48\textwidth]{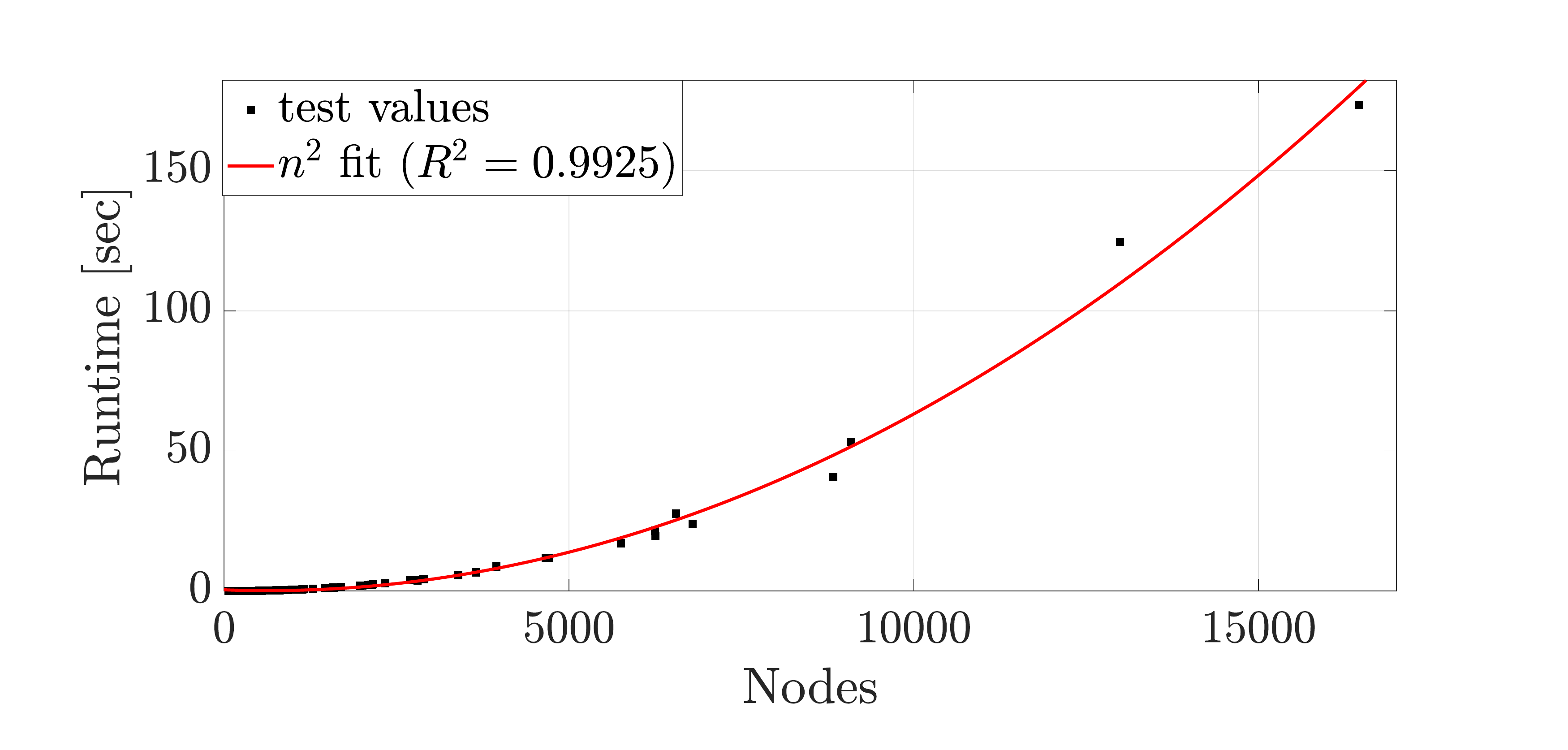}
\caption{Network creation time against node count.}
\label{fig:7}
\end{figure}

\begin{figure}[h]
\centering
\includegraphics[width=0.48\textwidth]{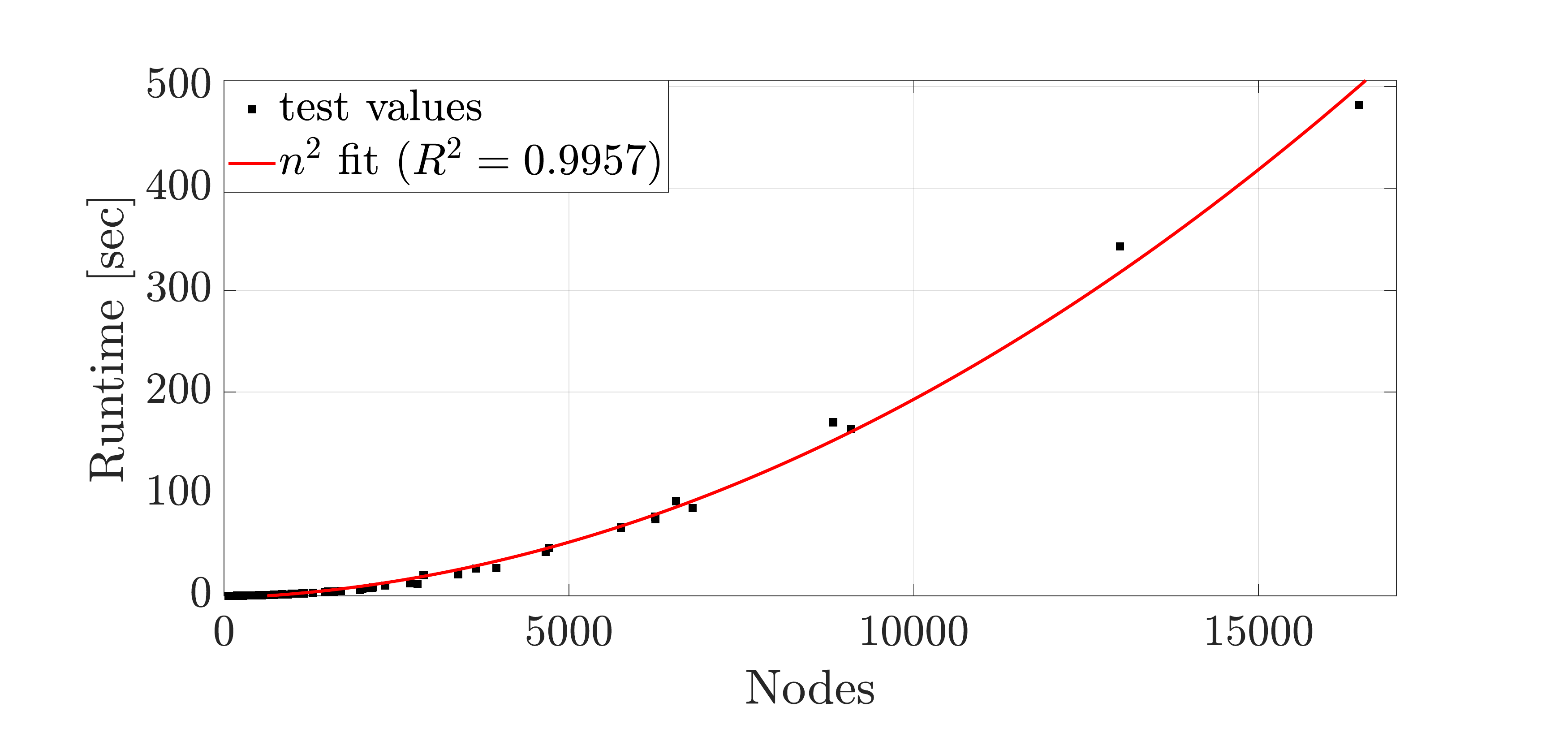}
\caption{SA runtime against node count.}
\label{fig:8}
\end{figure}

We define the trade surplus index (TSI) as the ratio of the objective value and the number of assigned vehicles in each instance. 
Figures \ref{fig:9} and \ref{fig:10} illustrate its relationship with the fleet coverage index (FCI), defined as the ratio of vehicles available against the number of vehicles required to serve all requests. An interesting feature of our approach (as shown in Figures \ref{fig:9} and \ref{fig:10}) is that the TSI is inversely proportional to the FCI for values of the latter between 0 and 1, and remains constant beyond that point.

This pattern can be explained by considering a scenario with 1 vehicle and 10 riders.
In this case, the node with the highest weight will be the solution in the MWIS problem. 
The addition of a new vehicle (with the same cost), assuming that it is included in the MWIS solution, will lead to a reduction in the average node weight.
This trend will persist with further increases in the size of the fleet, as riders with lower valuations are accommodated and gradually reduce the overall TSI. 
As such, once $FCI>1$ the TSI will on average remain constant, consistent with the notion of market equilibrium while supply increases beyond current demand levels.

\begin{figure}[h]
\centering
\includegraphics[width=0.48\textwidth]{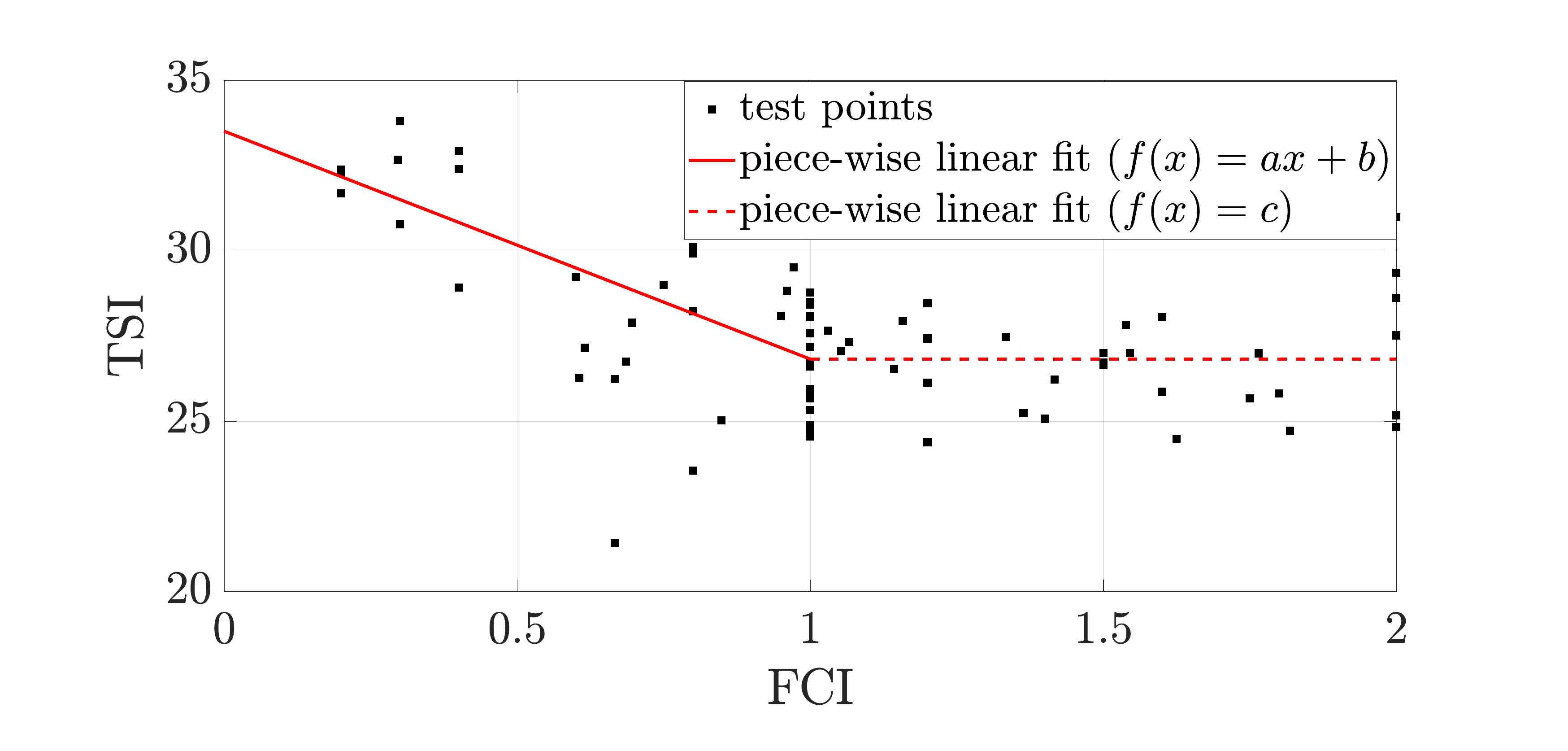}
\caption{Trade surplus per serving vehicle and fleet coverage.}
\label{fig:9}
\end{figure}

\begin{figure}[h]
\centering
\includegraphics[width=0.48\textwidth]{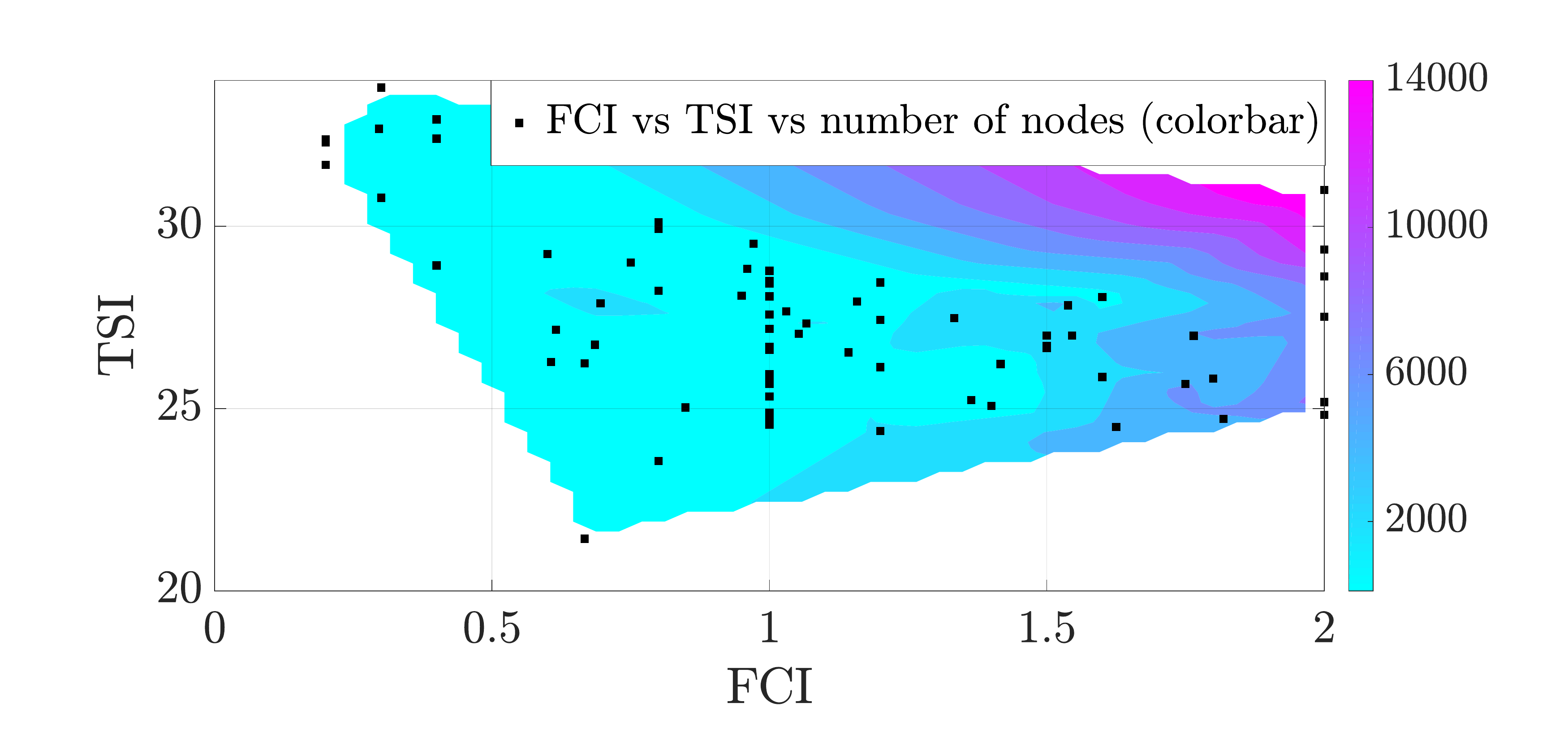}
\caption{Trade surplus per serving vehicle fleet coverage and node numbers}
\label{fig:10}
\end{figure}




\subsection{Practical Implementation} \label{sec:4.2}
To investigate the practical implementation of our proposed methodology, we analysed ride-sharing data provided by the Taxi and Limousine Commission (TLC) of New York City (NYC). Specifically, we exported the high volume for-hire vehicle trip records provided in \cite{TLC2019} and identified typical daily weekday trip count profiles which originate and terminate in the island of Manhattan NYC for the entirety of the ride-sharing market. 

By recording the MWIS node count for a varying request input, we were able to grasp the effect of ride requests on the problem size for an FCI equal to one (supply=demand), as shown in figure \ref{fig:11}. Using the identified runtime trends outlined in figures \ref{fig:7}, \ref{fig:8} and \ref{fig:11}, we compiled table \ref{tab:RuntimeTable}, which reports the time performance of our proposed methodology for varying request inputs. By assessing the request performance levels in table \ref{tab:RuntimeTable} we chose fifty requests as the practical limit in Manhattan, since our proposed methodology produces ride-sharing solutions approximately within one minute, which we regarded as acceptable. 

\begin{figure}[h]
\centering
\includegraphics[width=0.48\textwidth]{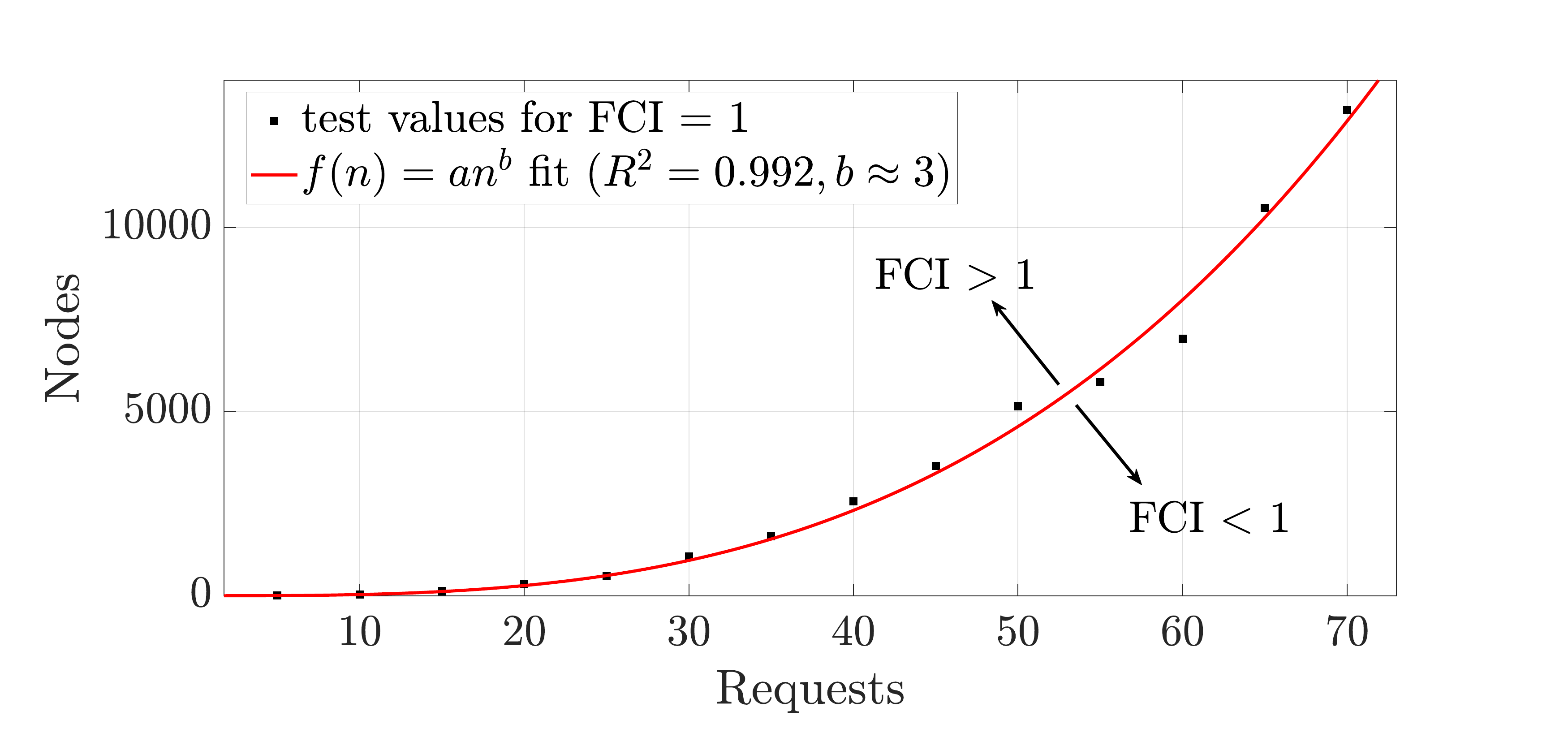}
\caption{MWIS node count for a varying request input in Manhattan, NYC.}
\label{fig:11}
\end{figure}

\begin{table}[]
\centering
\caption{MWIS instance node count and runtimes against varying request inputs.}
\label{tab:RuntimeTable}
\begin{tabular}{lllll}
\hline
Requests &
  Nodes &
  \begin{tabular}[c]{@{}l@{}}Network\\ runtime\\ {[}sec{]}\end{tabular} &
  \begin{tabular}[c]{@{}l@{}}SA\\ runtime\\ {[}sec{]}\end{tabular} &
  \begin{tabular}[c]{@{}l@{}}Total\\ runtime\\ {[}sec{]}\end{tabular} \\ \hline
40 & 2500 & 5  & 10  & 15  \\
45 & 3500 & 10 & 20  & 30  \\
50 & 5000 & 15 & 50  & 65  \\
55 & 6000 & 25 & 80  & 105 \\
60 & 8000 & 45 & 120 & 165 \\ \hline
\end{tabular}
\end{table}

By examining the typical per-minute shared ride count in Manhattan in figure \ref{fig:12}, we observe that the demand surpasses the cutoff of fifty requests only during three distinct demand peaks, specifically during the morning, afternoon and evening. As such, since the highest peak narrowly exceeds a hundred shared trip rides, our proposed methodology can be practically implemented during peak hours when demand for rides exceeds supply (FCI $\leq$ 1) with an assignment duration interval $\Delta$ of thirty seconds for the entire Manhattan shared ride market. 

\begin{figure}[h]
\centering
\includegraphics[width=0.48\textwidth]{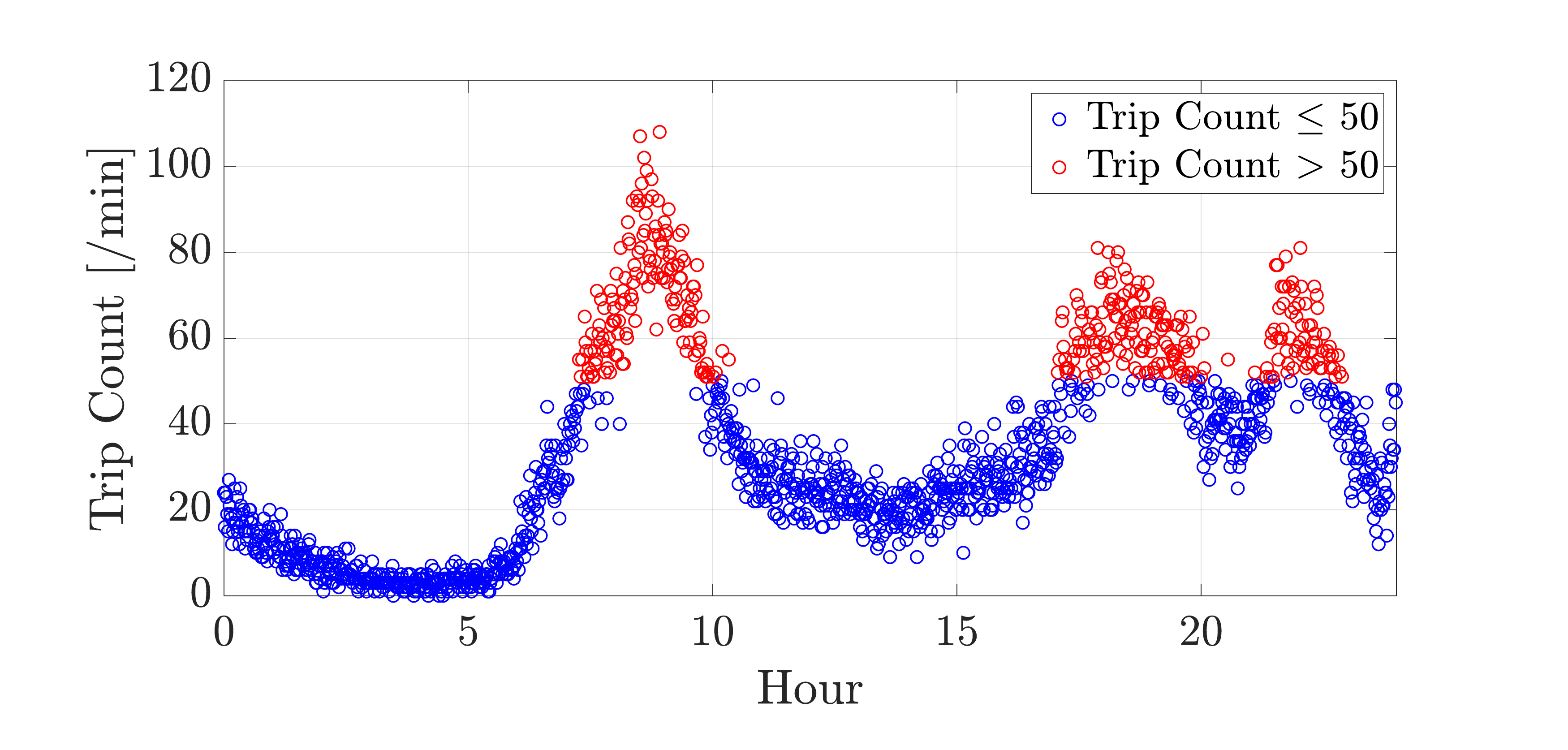}
\caption{Typical Weekday Per-Minute Shared-Trip Count in Manhattan, NYC.}
\label{fig:12}
\end{figure}

Nonetheless, the choice of a practical request cutoff value also depends on the error of the SA solution when compared to the exact solution. As observed in figure \ref{fig:11}, the number of nodes increases almost in a cubic rate with an increasing number of requests. Also, the percentage error increases approximately in a linear fashion with an increasing number of nodes, as observed in figure \ref{fig:3}. As such, an instance of 100 requests might have a comparable SA total utility value when split in three instead of two instances of 50 requests. For reference, by running BC instances of 50 requests for FCI$=1$, using the upper bound\footnote{An exact solution for such an instance size was prohibitive due to combinatorial explosion. As such, we used a long-run upper bound and best integer solution of the BC algorithm before termination.} and the best integer solution provided, we pinpoint the SA percentage error within $10\%-20\%$ of the exact solution. 

Our practical implementation recommendations above assume pervasiveness of autonomous vehicles similar to the levels of current conventional ride-sharing platforms. Nonetheless, the adoption rate of autonomous vehicles in commercial ride-sharing is a conjecture. As such, plausible scenarios could involve mixed fleets and custom rider requirements. Even so, our algorithm is still applicable for such customisation as one could screen any preferences in the pre-matching stage (Section \ref{sec:3.1}).

\section{Conclusion} \label{sec:5}

In this paper, we considered the problem of ride-sharing assignment and pricing in TNC platforms with autonomous vehicles. 
Our proposed assignment and pricing approach utilises a local search algorithm that solves a WDP MILP variant approximately in polynomial time by computing three-dimensional assignments to maximise trade surplus. 
By investigating the robustness of our proposed model, we derived a GFP auction interface which conveniently reduces to a stable three-dimensional assignment with minimal detours if riders report untruthful bids. We demonstrated the practicability of our proposed assignment and pricing method in a large urban setting such as Manhattan, NYC.

Our suggestions for future research in this area are twofold. 
First, we believe that both computational complexity and accuracy improvements are possible in exploring the breadth of meta-heuristics and machine learning algorithms in solving the proposed problem of ride-sharing auctions. 
Spatial clustering of requests, for example, could split much larger instances than the ones tested into parallel problems, which could be solved in a reasonable time, without compromising much of the efficiency of the algorithm. Secondly, agent-based modelling studies which focus on analysing heterogeneous bid behaviours in our proposed methodology (or a variant of it) could be useful. Such studies could produce large data-sets of solutions and aid in better assessing the effects of shill bidding in such a combinatorial auction setting. 

\bibliographystyle{IEEEtran}
\bibliography{references_file}

\begin{IEEEbiography}[{\includegraphics[width=1in,height=1.25in,clip,keepaspectratio]{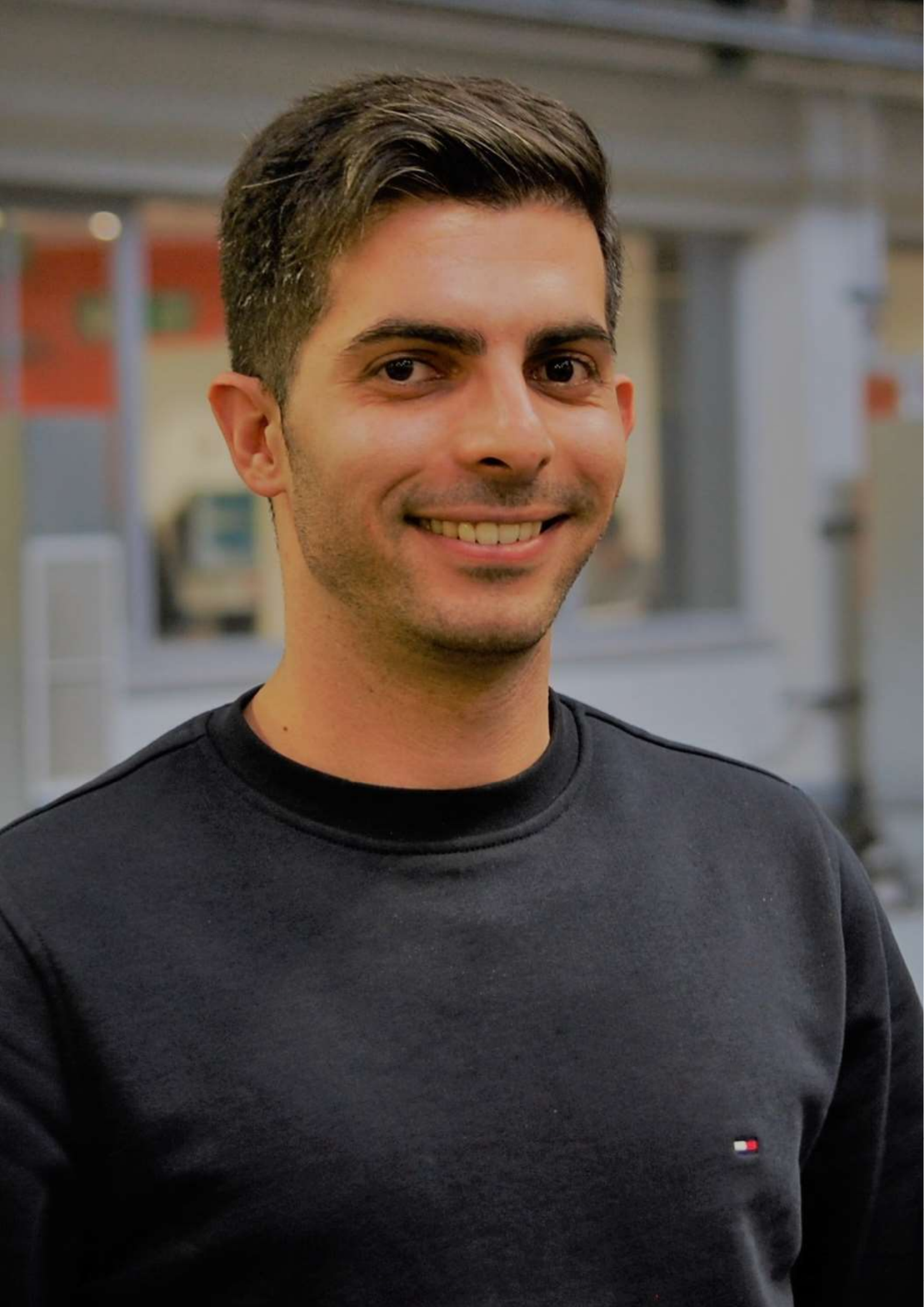}}]{Renos Karamanis}
received an M.Eng. degree in Civil and Environmental Engineering from Imperial College London in 2014. From 2014 to 2015 he worked as an engineering consultant for Mott MacDonald, a multi-disciplinary consultancy with headquarters in the United Kingdom. 

In 2015 he joined the Centre for Transport Studies in Imperial College London where he is currently a Ph.D. student. His research interests include developing operational research methods to improve the efficiency of pricing, assignment and resource allocation of autonomous ride-sourcing fleets, and simulation modelling of ride-sourcing operations to aid policy suggestions.
\end{IEEEbiography}

\begin{IEEEbiography}[{\includegraphics[width=1in,height=1.25in,clip,keepaspectratio]{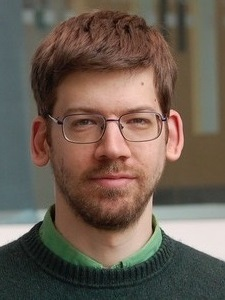}}]{Eleftherios Anastasiadis}
received a BSc in informatics and telecommunications from the University of Athens in 2011. He received an MSC and a PhD in theoretical computer science from the University of Liverpool in 2012 and 2017 respectively. 

In 2017 he worked as quality assurance engineer in Rosslyn Data Technologies Ltd. Since 2018 he is a postdoctoral researcher at the Transport Systems and Logistics Laboratory in the department of Civil and Environmental Engineering at ImperialCollege London. His research interests include approximation algorithms and mechanism design for network optimisation problems, and agent-based simulation for autonomous vehicle fleets. 
\end{IEEEbiography}

\begin{IEEEbiography}[{\includegraphics[width=1in,height=1.25in,clip,keepaspectratio]{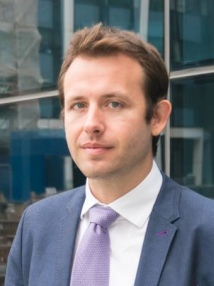}}]{Panagiotis Angeloudis}
is a Senior Lecturer and Director of the Transport Systems and Logistics Laboratory, part of the Centre for Transport Studies and the Department of Civil \& Environmental Engineering at Imperial College London. He received an MEng in Civil \& Environmental Engineering in 2005 and a PhD in Transport Operations in 2009, both from  Imperial College London. 

His research interests lie on the field of transport systems and networks operations, with a focus on the the efficient and reliable movement of people and goods across land, sea and water. He was recently appointed by the UK Department for Transport to the Expert Panel for Maritime 2050 and was a member of the UK Government Office of Science Future of Mobility review team. He is affiliated with the Centre for Systems Engineering and Innovation, the Institute for Security Science and Technology, the Grantham Institute and the Imperial Robotics Forum.
\end{IEEEbiography}

\begin{IEEEbiography}[{\includegraphics[width=1in,height=1.25in,clip,keepaspectratio]{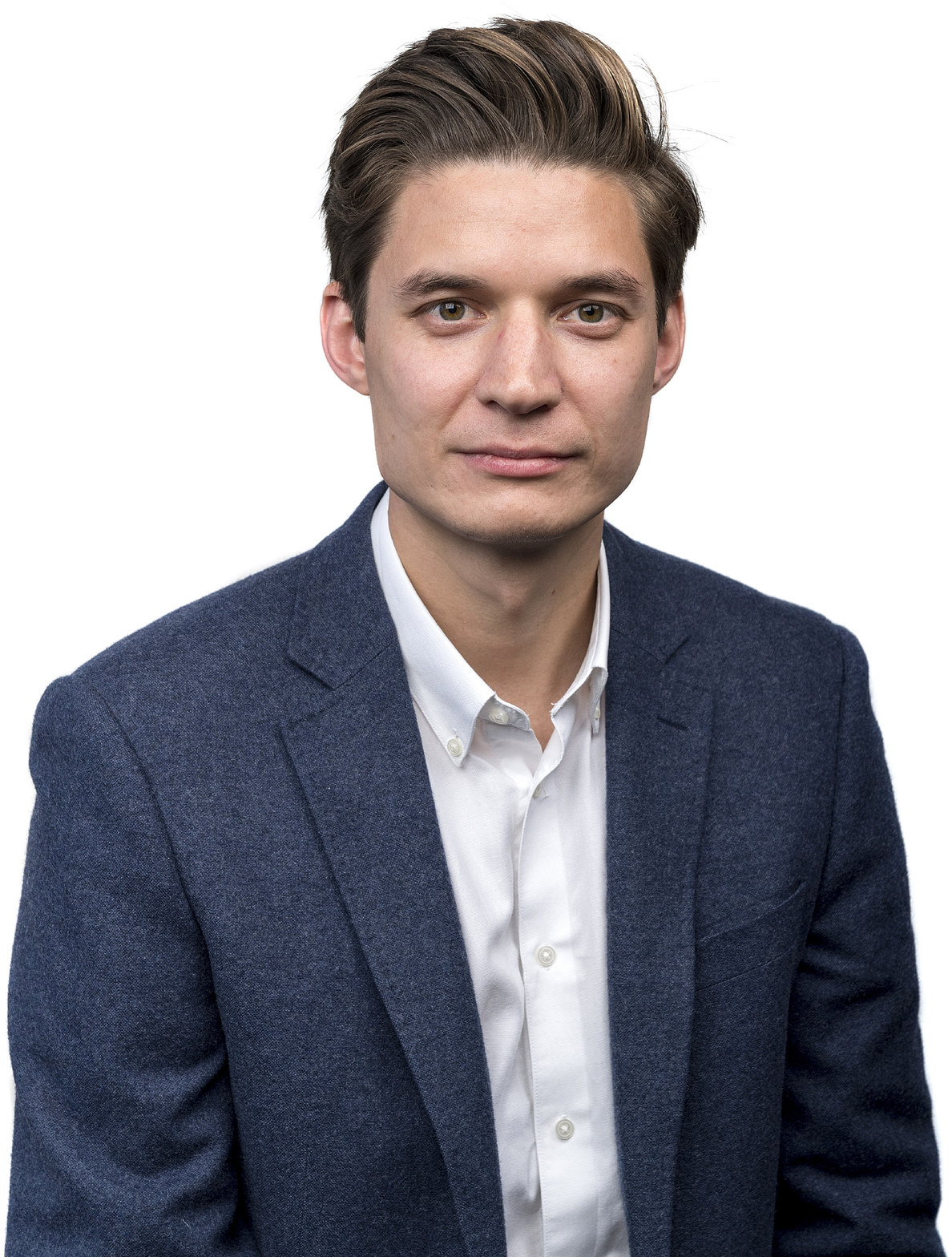}}]{Marc Stettler}
is a Senior Lecturer in Transport and the Environment in the Centre for Transport Studies and Director of the Transport \& Environment Laboratory. Prior to joining Imperial, Marc was a research associate in the Centre for Sustainable Road Freight and Energy Efficient Cities Initiative at the University of Cambridge, where he also completed his PhD.

His research aims to quantify and reduce environmental impacts from transport using a range of emissions measurement and modelling tools. Examples of recent research projects include: quantifying real-world vehicle emissions; using real-world vehicle emissions data to improve emissions models; evaluating economic and environmental benefits of Kinetic Energy Recovery Systems (KERS) for road freight; and quantifying aircraft emissions at airports. Marc is a member of the LoCITY ‘Policy, Procurement, Planning and Practice’ working group and the EQUA Air Quality Index Advisory Board. 
\end{IEEEbiography}

\end{document}